\documentclass[onecolumn]{IEEEtran}
\usepackage[latin1]{inputenc}
\usepackage[english]{babel}
\usepackage[margin=1in]{geometry}
\usepackage{xcolor}

\usepackage{blindtext}
\usepackage{amssymb}
\usepackage{hyperref}
\hypersetup{
    colorlinks=true,
    linkcolor=blue,
    filecolor=magenta,      
    urlcolor=cyan,
}
 
\usepackage{amsthm}
\usepackage[normalem]{ulem}
\usepackage{amsmath}
\usepackage{booktabs}
\usepackage{color}
\usepackage{tabularx}
\usepackage{array}
\usepackage{cite}
\usepackage{soul}
\usepackage{mathtools}
\usepackage{bm}
\usepackage{multirow}
\usepackage{makecell}
\usepackage{thmtools}
\usepackage{thm-restate}
\usepackage{graphicx}
\usepackage{subcaption}
\newtheorem{theorem}{Theorem}[section]
\newtheorem{corollary}[theorem]{Corollary}
\newtheorem{lemma}[theorem]{Lemma}
\newtheorem{proposition}[theorem]{Proposition}

\newtheorem{example}{Example}
\newtheorem{remark}[theorem]{Remark}

\usepackage{cleveref}

\newcommand\nc\newcommand
\nc{\cA}{\mathcal{A}}\nc{\cB}{\mathcal{B}}\nc{\cC}{\mathcal{C}}\nc{\cD}{\mathcal{D}}
\nc{\cE}{\mathcal{E}}\nc{\cF}{\mathcal{F}}\nc{\cG}{\mathcal{G}}\nc{\cH}{\mathcal{H}}
\nc{\cI}{\mathcal{I}}\nc{\cJ}{\mathcal{J}}\nc{\cK}{\mathcal{K}}\nc{\cL}{\mathcal{L}}
\nc{\cM}{\mathcal{M}}\nc{\cN}{\mathcal{N}}\nc{\cO}{\mathcal{O}}\nc{\cP}{\mathcal{P}}
\nc{\cQ}{\mathcal{Q}}\nc{\cR}{\mathcal{R}}\nc{\cS}{\mathcal{S}}\nc{\cT}{\mathcal{T}}
\nc{\cU}{\mathcal{U}}\nc{\cV}{\mathcal{V}}\nc{\cW}{\mathcal{W}}\nc{\cX}{\mathcal{X}}
\nc{\cY}{\mathcal{Y}}\nc{\cZ}{\mathcal{Z}}

\nc{\bba}{\mathbf{a}}\nc{\bbb}{\mathbf{b}}\nc{\bbc}{\mathbf{c}}\nc{\bbd}{\mathbf{d}}
\nc{\bbe}{\mathbf{e}}\nc{\bbf}{\mathbf{f}}\nc{\bbg}{\mathbf{g}}\nc{\bbh}{\mathbf{h}}
\nc{\bbi}{\mathbf{i}}\nc{\bbj}{\mathbf{j}}\nc{\bbk}{\mathbf{k}}\nc{\bbl}{\mathbf{l}}
\nc{\bbm}{\mathbf{m}}\nc{\bbn}{\mathbf{n}}\nc{\bbo}{\mathbf{o}}\nc{\bbp}{\mathbf{p}}
\nc{\bbq}{\mathbf{q}}\nc{\bbr}{\mathbf{r}}\nc{\bbs}{\boldsymbol{S}}\nc{\bbt}{\mathbf{t}}
\nc{\bbu}{\mathbf{u}}\nc{\bbv}{\mathbf{v}}\nc{\bbw}{\mathbf{w}}
\nc{\bbz}{\mathbf{z}}

\nc{\bbA}{\mathbf{A}}\nc{\bbB}{\mathbf{B}}\nc{\bbC}{\mathbf{C}}\nc{\bbD}{\mathbf{D}}
\nc{\bbE}{\mathbf{E}}\nc{\bbF}{\mathbf{F}}\nc{\bbG}{\mathbf{G}}\nc{\bbH}{\mathbf{H}}
\nc{\bbI}{\mathbf{I}}\nc{\bbJ}{\mathbf{J}}\nc{\bbK}{\mathbf{K}}\nc{\bbL}{\mathbf{L}}
\nc{\bbM}{\mathbf{M}}\nc{\bbN}{\mathbf{N}}\nc{\bbO}{\mathbf{O}}\nc{\bbP}{\mathbf{P}}
\nc{\bbQ}{\mathbf{Q}}\nc{\bbR}{\mathbf{R}}\nc{\bbS}{\mathbf{S}}\nc{\bbT}{\mathbf{T}}
\nc{\bbU}{\mathbf{U}}\nc{\bbV}{\mathbf{V}}\nc{\bbW}{\mathbf{W}}
\nc{\bbZ}{\mathbf{Z}}

\nc{\sA}{\mathsf{A}}\nc{\sB}{\mathsf{B}}\nc{\sC}{\mathsf{C}}\nc{\sD}{\mathsf{D}}
\nc{\sE}{\mathsf{E}}\nc{\sF}{\mathsf{F}}\nc{\sG}{\mathsf{G}}\nc{\sH}{\mathsf{H}}
\nc{\sI}{\mathsf{I}}\nc{\sJ}{\mathsf{J}}\nc{\sK}{\mathsf{K}}\nc{\sL}{\mathsf{L}}
\nc{\sM}{\mathsf{M}}\nc{\sN}{\mathsf{N}}\nc{\sO}{\mathsf{O}}\nc{\sP}{\mathsf{P}}
\nc{\sQ}{\mathsf{Q}}\nc{\sR}{\mathsf{R}}\nc{\sS}{\mathsf{S}}\nc{\sT}{\mathsf{T}}
\nc{\sU}{\mathsf{U}}\nc{\sV}{\mathsf{V}}\nc{\sW}{\mathsf{W}}\nc{\sX}{\mathsf{X}}
\nc{\sY}{\mathsf{Y}}\nc{\sZ}{\mathsf{Z}}

\newcommand{\F}{{\Bbb F}}
\newcommand{\Fq}{{{\Bbb F}}_{\!q}}

\newcommand{\E}{{\Bbb E}}

\renewcommand{\Pr}{\mathrm{Pr}}
\newcommand{\codim}{\operatorname{codim}}

\DeclareMathOperator{\spn}{Span}
\DeclareMathOperator{\rank}{rank}
\newcommand{\gv}[1]{[#1]_q}

\newcommand{\qbin}[2]{\genfrac{[}{]}{0pt}{}{#1}{#2}_q}

\newcommand{\bone}{\mathbf{1}} 

\newcommand{\bra}[1]{[#1]_q}

\newcommand{\mathset}[1]{\left\{#1\right\}}

\newcommand{\ceilenv}[1]{\left\lceil #1 \right\rceil}
\newcommand{\floorenv}[1]{\left\lfloor #1 \right\rfloor}
\newcommand{\parenv}[1]{\left( #1 \right)}
\newcommand{\sparenv}[1]{\left[ #1 \right]}

\newcommand{\bal}[1]{\begin{align}\label{#1}}
\newcommand{\eal}{\end{align}}

\renewcommand{\le}{\leqslant}
\renewcommand{\leq}{\leqslant}
\renewcommand{\ge}{\geqslant}
\renewcommand{\geq}{\geqslant}

\theoremstyle{definition}
\newtheorem{definition}{Definition}[section]
\newtheorem{problem}{Problem}

\usepackage{comment}

\definecolor{revisionyellow}{RGB}{180,140,0}

\title{The Coverage Depth Problem in Distributed DNA Data Storage}
\author{Xiangliang~Kong, Ohad~Elishco, Chen~Wang, and Tolga~M.~Duman%
\thanks{This work was supported in part by the National Key R\&D Program of China under Grant 2025YFA1017200 and by the ERC Advanced Grant 101054904: TRANCIDS. Views and opinions expressed are however those of the authors only and do not necessarily reflect those of the European Union or the European Research Council Executive Agency. Neither the European Union nor the granting authority can be held responsible for them. The research of C. Wang was supported in part at the Technion by a fellowship from the Lady Davis Foundation.} 
\thanks{This research was partially supported by the Israel Science Foundation (Grant No. 1789/23).}
\thanks{X. Kong (rongxlkong@gmail.com) is with the State Key Laboratory of Mathematical Sciences, Academy of Mathematics and Systems Science, Chinese Academy of Sciences, Beijing 100190, China. O. Elishco (elishco@gmail.com) is with the School of Electrical and Computer Engineering, Ben-Gurion University of the Negev, Beer Sheva 8410501, Israel. C. Wang (cwang@campus.technion.ac.il) is with the Department of Computer Science, Technion -- Israel Institute of Technology, Haifa 3200003, Israel. Tolga M. Duman (duman@ee.bilkent.edu.tr) is with the Department of Electrical and Electronics Engineering, Bilkent University, Ankara 06800, Turkey.}
}
\date{\today}

\begin{document}

\maketitle

\begin{abstract}
Random sampling in DNA sequencing produces repeated reads, increasing retrieval latency and sequencing cost. We study the coverage-depth problem for full-message recovery in distributed DNA storage under noiseless uniform sampling, where strands are partitioned among $M$ containers and one strand is independently sampled with replacement from each container per round. For arbitrary linear codes and ordered partitions, we derive exact formulas for the recovery-time distribution and expectation. We prove that MDS codes, whenever they exist, are optimal for every fixed partition, and establish a universal lower bound on the expected total read cost together with its equality conditions. For MDS codes, we identify container-size regimes that yield genuine savings in total reads and regimes that provide only parallelism without changing the asymptotic sequencing cost. For simplex codes, we prove that the $q$-ary simplex code is, up to isomorphism, the unique single-container minimizer among codes with the same parameters, resolving a recent conjecture by Bertuzzo, Ravagnani, and Yaakobi. We further construct a partition attaining the minimum total read cost and derive bounds for intermediate and balanced partitions. These results clarify when distributed sampling reduces latency alone and when it also reduces sequencing cost.
\end{abstract}

\section{Introduction}
\label{sec:introduction}

DNA-based data storage has emerged as a promising technology for long-term digital archiving. DNA offers extremely high information density, long-term chemical stability, and low maintenance requirements, and a number of proof-of-concept systems have demonstrated the feasibility of storing and retrieving digital information using synthetic DNA molecules \cite{ChurchEtAl12,GoldmanEtAl13,YazdiEtAl15,OrganickEtAl18}. The coding-theoretic and information-theoretic challenges arising in such systems have consequently attracted much attention, for example,  see the surveys \cite{ShomoronyHeckel22,SabaryEtAl24} and the references therein.

A typical DNA storage pipeline consists of encoding, synthesis, physical storage, sequencing, and decoding. The user data are first encoded into a collection of short DNA sequences, commonly referred to as \emph{oligos} or \emph{strands}. These designed strands are synthesized, amplified into many physical copies, and stored together in one or more containers. To retrieve the information, DNA molecules are sampled and sequenced, producing an unordered multiset of reads. Indices embedded in the strands may identify the designed strand from which a read originates, but they do not make the molecules directly addressable: the sequencing process still samples molecules from the physical pool, and repeated observations of the same designed strand are unavoidable.

This random-sampling feature distinguishes DNA storage from conventional storage devices. Even in an idealized noiseless model, a sequencer may repeatedly return strands that have already been observed, while other strands remain unseen. Consequently, the number of reads required for decoding can be substantially larger than the number of information strands from the  message. In physical terms, the \emph{coverage depth} is the number of sequenced reads normalized by the number of designed strands \cite{HeckelEtAl19}. Since the number of designed strands is fixed for a given code, minimizing coverage depth is equivalent to minimizing the total number of reads. The latter is directly related to both sequencing cost and recovery latency.

The resulting sampling problem is closely connected to the classical coupon collector problem \cite{ErdosRenyi61,Feller68,FlajoletEtAl92,BonehHofri97}. In the uncoded setting, every designed strand is a coupon type, and decoding requires observing all coupon types. Error-correcting codes change the stopping condition: rather than observing every coded strand, it is enough to observe a subset that determines the original message. Thus, coding replaces the coupon collector's ``collect all coupons'' criterion by a decodability criterion.

More precisely, consider a $k$-dimensional linear code $\mathcal C\subseteq\mathbb F_q^n$ with a rank-$k$ generator matrix $\mathbf G$. Each column of $\mathbf G$ corresponds to one coded DNA strand. In the single-container noiseless model, a read selects one column uniformly at random with replacement, and the message can be fully recovered when the sampled columns span $\mathbb F_q^k$. Bar-Lev \emph{et al.}~\cite{BSGY25} initiated the systematic theoretical study of this stopping time, referred to as the \emph{coverage-depth problem in DNA data storage}. They investigated its dependence on the sampling distribution, the outer code, and the strand-retrieval procedure. In particular, for uniform noiseless sampling, they proved that an $[n,k]_q$ MDS code, whenever it exists, minimizes the expected number of reads among $[n,k]_q$ linear codes. 


\subsection{Related Work}

Several extensions of the coverage-depth problem in DNA data storage have recently been investigated. Besides full-message recovery, Bar-Lev \emph{et al.}~\cite{BSGY25} introduced the singleton random-access problem, in which the user requests one prescribed information strand rather than the entire message. 
The combinatorial and geometric aspects of random access have subsequently been studied in \cite{GBRY25,GruicaMontanucciZullo24,BoruchovskyEtAl25}, which showed that code optimality for full-message recovery does not imply optimality for random-access objectives.

Other related models include recovering groups of strands representing a common information unit \cite{AbrahamEtAl24}, optimizing the trade-off between synthesis and sequencing costs through practical coding architectures \cite{ChandakEtAl19}, analyzing sequencing coverage in combinatorial shortmer-based storage systems \cite{PreussEtAl24}, and studying coverage and decoding probabilities for composite DNA letters \cite{CohenYaakobi24}. These models modify either the physical representation of information or the decoding objective, and therefore lead to stopping rules different from the full-rank condition considered in this paper.

When MDS codes do not exist over the prescribed field, the performance of structured small-field codes becomes a central question. Bertuzzo, Ravagnani, and Yaakobi~\cite{BRY26} developed algebraic and combinatorial methods for computing the single-container coverage depth through information-set enumeration, duality and weight distributions. They obtained explicit formulas for several classical code families, including simplex, Hamming, ternary Golay, extended ternary Golay, and first-order Reed--Muller codes. Based on the resulting formula for simplex codes and computational evidence, they conjectured that the $q$-ary simplex code minimizes the coverage depth among all codes having the same length and dimension. One of the results of the present paper resolves this conjecture in the affirmative and establishes the corresponding uniqueness statement.

\subsection{Distributed Full-Message Recovery}

Physical separation of coded data across multiple storage units is a standard mechanism for providing parallel access and fault tolerance. In classical distributed storage systems, a file is encoded across several nodes, and a central question is how much information must be downloaded from the surviving nodes to repair a failed one. The repair-bandwidth problem was introduced by Dimakis \emph{et al.}~\cite{DimakisEtAl10} and has since developed into an extensive theory of regenerating and related storage codes; see \cite{RamkumarEtAl22} for a survey.

DNA containers, however, are passive physical objects rather than computational nodes. A container cannot evaluate a prescribed function of its contents and transmit the result. Instead, information is obtained by sequencing, and every read returns a randomly sampled strand from the container. Thus, even if the same erasure code is used, the download process in DNA storage is fundamentally different from that in a conventional distributed storage system.

Levy, Con, Yaakobi, and Kiah~\cite{LCYK26} recently initiated the study of DNA-based distributed storage from the viewpoint of container failures. Their goal is to reconstruct the contents of a failed container by sequencing the surviving containers, and they analyzed the associated recovery time for scalar MDS and MDS array codes.

The objective of the present paper is complementary. We consider the problem of recovering the \emph{entire message} when the coded strands are distributed among several containers. This question arises even in the absence of container failures and remains necessary after a failed container has been repaired: the ultimate readout operation must still accumulate enough information to decode the complete message. The two problems also have different stopping rules. Failed-container recovery asks whether a prescribed missing block is determined by the observations from the surviving containers, whereas full-message recovery asks whether all observed coded strands jointly span the entire $k$-dimensional message space.

Specifically, we partition the $n$ columns of $\mathbf G$ into $M$ labeled containers according to an ordered partition $\mathcal P$. In every sequencing round, one column is sampled independently and uniformly, with replacement across rounds, from each container. Let $\tau(\mathbf G,\mathcal P)$ be the first round in which the sampled columns span $\mathbb F_q^k$. Since parallel sequencing can reduce latency without increasing the total number of reads, we distinguish two performance metrics and consider them separately: the expected recovery time in rounds, $\mathbb E[\tau(\mathbf G,\mathcal P)]$, and the expected total number of reads, $M\mathbb E[\tau(\mathbf G,\mathcal P)]$.

This formulation gives rise to the \emph{distributed coverage depth problem}, which consists of the following three interconnected questions: (1) For a fixed ordered $M$-partition, which code minimizes the expected recovery time? (2) For a fixed code and fixed $M$, how to partition its columns to achieve the smallest expected recovery time? (3) Under practical capacity constraints (such as balanced container sizes), which codes have a genuine reduction in the total sequencing cost (total number of reads), rather than only a reduction in latency (measured in recovery rounds)?

In this paper, we address the three questions listed above. We study the expected recovery time for both general linear codes and specific code families, providing explicit estimations of the expected recovery rounds. Furthermore, we analyze the extent to which single-container overhead can be reduced in distributed settings.


\subsection{Main Contributions}

The principal contributions of this work are summarized below.

\begin{itemize}

\item By decomposing the recovery process into successive stages and using $q$-binomial M\"obius inversion on the subspaces of $\mathbb{F}_q^k$, we derive an exact formula for $\mathbb E[\tau(\mathbf{G},\mathcal{P})]$ for an arbitrary rank-$k$ linear code with generator matrix $\mathbf{G}$ and an arbitrary ordered partition $\mathcal{P}$.

\item We show that whenever an $[n,k]_q$ MDS code exists, it is optimal for every fixed partition, both in terms of expected recovery time and total read cost. Furthermore, we establish a general lower bound on the expected number of recovery rounds, hence on the expected total number of reads, and we characterize exactly when $\tau(\bbG,\cP)=1$ a.s. In particular, for $M=k$ the bound is attained by the staircase partition of any linear code, so that $\min_{M,\cP} M\E[\tau(\bbG,\cP)]=k$ for every $[n,k]_q$ code. We then exhibit one-round-perfect layouts with additional structure for MDS, Hamming, first-order Reed-Muller and simplex codes. 

\item We consider the distributed coverage depth problem for MDS codes. When every container has a fixed size $s$, we identify explicit parameter regimes in which the recovery time concentrates on an integer $r$. Within part of this regime, the resulting total read cost is lower than that of the single-container case by a quantity linear in $k$. For containers of size $k^\varepsilon$ with $0<\varepsilon< 1$, we derive asymptotic formulas for both the expected recovery time and the total number of reads, showing that distribution provides parallel access without changing the asymptotic sequencing cost.

\item We also consider the distributed coverage depth problem for simplex codes. We first show that the $q$-ary simplex code is, up to isomorphism, the unique length-$[k]_q$, dimension-$k$ code that minimizes the single-container coverage depth, thereby resolving \cite[Conjecture~3.2]{BRY26}. Then, we provide an explicit $k$-partition that achieves the optimal total number of reads for recovery. Moreover, for $1<M<k$ we show that an intermediate number of containers cannot reduce the total read cost whenever $M\nmid k$ and $M\nmid k+1$, and also whenever $M\mid k+1$ provided $q\geq 3$ or $k=3$. The remaining regimes, namely $M\mid k$ for every $q$, and $M\mid k+1$ for $q=2$ and $k\geq 4$, are settled for cases $(q,k,M)=(2,5,2)$ and $(3,4,2)$.
\end{itemize}

Taken together, these results reveal that distributed sampling has two fundamentally distinct effects: parallelism reduces recovery latency, whereas a reduction in total sequencing cost requires the layout $(\mathbf{G},\mathcal{P})$ to suppress repeated or rank-deficient observations. The former effect is largely architectural; the latter depends delicately on the interaction between the algebraic properties of the code and the combinatorial structure of the partition.

\subsection{Organization}

The remainder of this paper is organized as follows. Section~\ref{sec: problem setting} introduces the distributed coverage depth model along with its metrics and foundational notions; it derives exact formulas for $\mathbb E[\tau(\mathbf G,\mathcal P)]$, proves the optimality of MDS codes, and develops general lower bounds on the total number of reads and the theory of one-round-perfect layouts. Section~\ref{sec: MDS distributed coverage} specializes this framework to MDS codes and investigates the associated distributed coupon-collector problem in constant-size, polynomial-size, and two-container regimes. Section~\ref{sec: simplex distributed coverage} focuses on simplex codes, establishing their unique single-container optimality and analyzing staircase, balanced, and random balanced layouts. Finally, Section~\ref{sec:conclusion} concludes the paper and outlines several open questions.

\section{Problem Setting and General Results}\label{sec: problem setting}

\subsection{Distributed Coverage-Depth Problem for DNA Storage}\label{subsec: problem statement}


Let $q$ be a prime power, $\Fq$ be the field of order $q$, and $k\ge2$. For two spaces $U$ and $V$, we write $U\le V$ if $U$ is a subspace of $V$, and $U\lneq V$ if $U$ is a proper subspace of $V$. 

We first recall the coverage depth problem introduced in \cite{BSGY25}. 

\begin{problem}[The DNA data coverage depth problem, \cite{BSGY25}]\label{def: coverage depth}
Let $\mathcal{C}\subseteq \mathbb{F}_q^{n}$ be a $k$-dimensional code of length $n$ defined by a rank-$k$ generator matrix $\mathbf{G} \in \mathbb{F}_q^{k \times n}$ with no zero columns, whose columns $\mathbf{g}_1, \dots, \mathbf{g}_n$ form a multiset of $n$ vectors. Suppose we draw columns independently and uniformly with replacement. Let
\begin{equation}\label{eq: def CD-matrix}
\tau \triangleq \min\left\{t \in \mathbb{Z}^+ : \dim(\spn( \mathbf{g}_{j_1}, \dots, \mathbf{g}_{j_t})) = k\right\},
\end{equation}
where $\mathbf{g}_{j_1}, \dots, \mathbf{g}_{j_t}$ are the columns drawn in rounds $1, \dots, t$. Then, the coverage depth problem is to determine the value of $\mathbb{E}[\tau]$.
\end{problem}

Next, we define the \textbf{coverage depth problem in distributed DNA data storage}. 

Given $k$ different information DNA strands $(x_1,x_2,\dots,x_k)$, they are encoded into $n$ coded DNA strands $(c_1,c_2,\ldots,c_n)$ using a code $\mathcal{C}$ with generator matrix $\mathbf{G}\in \mathbb{F}_q^{k\times n}$, and stored in $M$ different containers according to an ordered $M$-partition $\mathcal{P}=(P_1,P_2,\ldots,P_M)$ of $[n]$ with non-empty parts. Unless explicitly stated otherwise, $k\ge 2$, the generator matrix $\mathbf{G}$ has rank $k$ and no zero columns. Moreover, repeated columns and nonzero scalar multiples of columns are allowed and are counted with their coordinate multiplicities.

Write $\mathbf{G}$ in column-wise form as
$$
\mathbf{G}=(\mathbf{g}_1,\mathbf{g}_2,\ldots,\mathbf{g}_{n}),
$$
where $\mathbf{g}_i\in \mathbb{F}_q^{k}$ denotes the $i$-th column. Let $n_i=|P_i|$ for each $1\le i\le M$, so that $\sum_{i\in[M]} n_i=n$. For each $i\in[M]$, define
$$
\mathbf{G}_i\triangleq \mathbf{G}|_{P_i},
$$
where $\mathbf{G}|_{P_i}$ denotes the $k\times n_i$ submatrix of $\mathbf{G}$ obtained by restricting $\mathbf{G}$ to the columns indexed by $P_i$. Then, 
$$
\mathbf{c}_i=(x_1,\ldots,x_k)\cdot \mathbf{G}_{i}=(c_{i,1},\ldots,c_{i,n_i})
$$
denotes the vector of coded strands stored in the $i$-th container.

\begin{problem}[The Coverage Depth Problem in Distributed DNA Data Storage]
Let the code $\mathcal{C}$ with generator matrix $\mathbf{G}$ and the ordered partition $\mathcal{P}$ be as defined above. In each round, one column is drawn independently and uniformly at random with replacement from each submatrix $\mathbf{G}_i$. For each $i \in [M]$, let $j_{i,t}$ denote the index of the column drawn from the $i$-th container at round $t$, and define the set of distinct column indices observed by the end of round $t$ as
\begin{equation}\label{eq: column index set by round t}
    J_t \triangleq \{j_{i,s} : i \in [M],\, 1 \le s \le t\},
\end{equation}
with $J_0 \triangleq \varnothing$. Let 
\begin{equation}\label{eq: column space by round t}
    V_t \triangleq \spn\{\mathbf{g}_j : j \in J_t\}.
\end{equation}
with $V_0 \triangleq \{\mathbf{0}\}$. Define the stopping time
\begin{equation}\label{eq:def_tau_G_P}
    \tau(\mathbf{G},\mathcal{P}) \triangleq \min\left\{
        t \in \mathbb{Z}^+ :
        \mathrm{span}\left( \{\mathbf{g}_j : j \in J_t\} \right) = \mathbb{F}_q^k
    \right\}.
\end{equation}
The distributed coverage depth problem associated with the pair $(\mathbf{G},\mathcal{P})$ is then to determine the expected value $\mathbb{E}\left[\tau(\mathbf{G},\mathcal{P})\right]$.
\end{problem}

\begin{remark}\label{rmk: synchronous}
The model is synchronous: exactly one read is taken from every container in every round, and the final round is charged $M$ reads even if a single read would have sufficed to complete the span. Hence, $M\E[\tau(\bbG,\cP)]\geq M\ceilenv{k/M}$, and part of the measure overhead is an artifact of this accounting rather than of the sampling. An asynchronous variant, in which a total read budget is allocated is discussed as an open problem in Section~\ref{sec:conclusion}.  
\end{remark}

Throughout the paper, for the single-container layout $\mathcal{P}^{(1)} \triangleq ([n])$, we write $\tau^{(1)}(\mathbf{G}) \triangleq \tau(\mathbf{G},\mathcal{P}^{(1)})$ for simplicity. Moreover, by~\cite[Proposition 2.1]{BRY26}, for any linear code $\mathcal{C}$, the expectation $\mathbb{E}[\tau^{(1)}(\mathbf{G})]$ is invariant to the choice of the generator matrix $\mathbf{G}$. Hence, we will use the notations $\mathbb{E}[\tau^{(1)}(\mathbf{G})]$ and $\mathbb{E}[\mathcal{C}]$ interchangeably.

Note that in DNA storage, the sequencing cost depends primarily on the total number of reads required during the information recovery process. In the single-container model ($M=1$), where one read is performed per round, the total number of reads simply equals the number of rounds needed for recovery. However, in the distributed model with $M \ge 2$, this value equals the number of containers times the number of rounds. In this paper, beyond determining $\mathbb{E}\left[\tau(\mathbf{G},\mathcal{P})\right]$ for various codes and partitions, we also evaluate the expected total number of reads, $M\mathbb{E}\left[\tau(\mathbf{G},\mathcal{P})\right]$, during the recovery process.

According to its definition, the ordered partition $\mathcal{P}$ can exert a substantial influence on the value of $\mathbb{E}[\tau(\mathbf{G},\mathcal{P})]$ for fixed $\mathbf{G}$ and $M$. Moreover, among all possible partitions, the case where all parts have the same size is of particular interest to storage system engineers, since the corresponding containers consequently have identical capacities and can be treated uniformly in the system. We call such a partition \emph{balanced} and present its formal definition as follows.

\begin{definition}[Balanced $M$-Partition]\label{def:balanced-layout-unified}
An ordered $M$-partition $\mathcal{P}=(P_1, \dots, P_M)$ is called \emph{balanced} if $|P_1| = \cdots = |P_M| = n/M$, which requires that $M \mid n$.
\end{definition}

Note that $\tau(\mathbf{G},\mathcal{P}) \ge 1$ holds for every rank-$k$ matrix $\mathbf{G}$ and every ordered $M$-partition $\mathcal{P}$ for any integer $M$. To characterize one of the extremal cases where $\tau(\mathbf{G},\mathcal{P}) = 1$ a.s. (almost surely), we introduce the following notion.

\begin{definition}[One-Round Decodability and One-Round Perfection]
\label{def:one-round-perfect-unified}
A coded layout $(\mathbf{G},\mathcal{P})$, where $\mathbf{G}$ is a generator matrix of an $[n,k]_q$ linear code and $\mathcal{P}$ is an ordered $M$-partition of $[n]$, is called \emph{one-round-decodable} if every $M$-tuple $(\mathbf{g}_{i_1}, \mathbf{g}_{i_2}, \dots, \mathbf{g}_{i_M})$ with ${i_j} \in P_j$ for all $j \in [M]$ spans $\mathbb{F}_q^k$, or equivalently, if $\tau(\mathbf{G},\mathcal{P}) = 1$ a.s. If, in addition, $M = k$, then the layout $(\mathbf{G},\mathcal{P})$ (or simply $\mathcal{P}$, when $\mathbf{G}$ is clear from context) is called \emph{one-round-perfect}.
\end{definition}

To characterize the advantages of distributed storage, we introduce the following notions.

\begin{definition}[Overhead and Distributed Saving]
\label{def:overhead-saving-unified}
We define the single-container recovery overhead for an $[n,k]_q$ linear code $\mathcal{C}$ with generator matrix $\mathbf{G}$ as
\begin{equation}\label{eq:def-overhead}
    \mathrm{Ov}(\mathcal{C}) \triangleq \mathbb{E}[\mathcal{C}] - k.
\end{equation}
The saving of an $M$-container layout $(\mathbf{G},\mathcal{P})$ relative to the single-container model, in terms of the expected total number of reads, is defined as $\mathbb{E}[\tau^{(1)}(\mathbf{G})] - M\mathbb{E}[\tau(\mathbf{G},\mathcal{P})]$. The layout achieves \emph{full saving} if $\mathbb{E}[\tau^{(1)}(\mathbf{G})] - M\mathbb{E}[\tau(\mathbf{G},\mathcal{P})] = \mathrm{Ov}(\mathcal{C})$, or equivalently, if $M\mathbb{E}[\tau(\mathbf{G},\mathcal{P})] = k$.
\end{definition}

\subsection{Exact Formulas for Distributed Coverage Depth}\label{subsec: exact formula}

For an integer $m\ge 0$, write $[m]_q = \frac{q^m-1}{q-1} = \sum_{i=1}^{m}q^{i-1}$ (with $[0]_q = 0$) for the number of points in the $(m-1)$-dimensional projective space $\mathrm{PG}(m-1,q)$, and write
$$
\begin{bmatrix} m \\ j \end{bmatrix}_q = \prod_{i=0}^{j-1} \frac{q^{m-i}-1}{q^{j-i}-1}
$$
for the number of $j$-dimensional subspaces of $\mathbb{F}_q^m$. Let $\mathbb{P} \triangleq \mathrm{PG}(k-1,q)$ denote the set of projective points in $\mathbb{F}_q^k$.

For a subspace $U \leq \mathbb{F}_q^{k}$, a generator matrix $\mathbf{G}\in \mathbb{F}_q^{k\times n}$, and an ordered $M$-partition $\mathcal{P}$ of $[n]$, let $m_{i}^{U}=|\mathbf{G}_i\cap U|$ for each $i\in [M]$, and define
\begin{equation}\label{eq:rho}
\rho_U=\prod_{i=1}^{M}\frac{m_i^{U}}{n_i},
\end{equation}
where $\mathbf{G}_i\cap U$ denotes the set of columns in $\mathbf{G}_i$ that lie in $U$. Because columns are drawn independently across different containers and uniformly within each container, $\rho_U$ is the probability that all $M$ columns sampled in a single round lie in $U$. Moreover, when $M=1$, we simply write $\pi(U)=m_1^{U}/n$ instead of $\rho_U$.


We next establish the exact formula for $\mathbb{E}[\tau(\mathbf{G},\mathcal{P})]$ given general $\mathbf{G}$ and $\mathcal{P}$. To this end, we first present the following lemma, which is a standard consequence of the M\"obius inversion formula.

\begin{lemma}\label{lem:qie}
For any $A\subseteq\Fq^{k}$, we have
$$
\sum_{c\ge1}(-1)^{c-1}q^{\binom{c}{2}}\bigl|\{U\le \Fq^k:\codim(U)=c,\ A\subseteq U\}\bigr| = \bone_{[\spn(A)\neq\Fq^{k}]}.
$$
\end{lemma}

\begin{IEEEproof}
Let $a=\dim(\spn(A))$. Clearly, a subspace $U$ contains $A$ if and only if it contains $\spn(A)$.

Instead of counting in the space $\Fq^k$, we count in the quotient space $\Fq^k/\spn(A)$, which has dimension $k-a$. The subspaces of $\Fq^k$ of codimension $c$ that contain $\spn(A)$ are in one-to-one correspondence with the codimension-$c$ subspaces of this quotient space. The number of such subspaces is $\qbin{k-a}{c}$. Therefore, the sum in the statement can be rewritten as
\begin{equation}\label{eq:qie1}
    \sum_{c\ge 1}
    (-1)^{c-1}q^{\binom{c}{2}}
    \bigl|\{U\le \Fq^k:\codim(U)=c,\ A\subseteq U\}\bigr|
    =
    \sum_{c\ge 1}
    (-1)^{c-1}q^{\binom{c}{2}}
    \qbin{k-a}{c}.
\end{equation}

To evaluate the right-hand side, we use the $q$-binomial theorem
\cite[Chapter 3]{JM17} in the form
$$
\sum_{c\ge0}(-1)^{c}q^{\binom{c}{2}}\qbin{k-a}{c} = \prod_{i=0}^{k-a-1}(1-q^{i}).
$$
We now evaluate this product according to whether $A$ spans the full space.

\begin{itemize}
    \item If $a<k$, i.e., $\spn(A)\neq\Fq^k$, then the factor corresponding to
    $i=0$ appears in the product. Since $1-q^0=0$, the product is zero.

    \item If $a=k$, i.e., $\spn(A)=\Fq^k$, then the product is empty and hence
    equals $1$.
\end{itemize}
Thus,
$$
\sum_{c\ge0}(-1)^{c}q^{\binom{c}{2}}\qbin{k-a}{c} = \bone_{[\spn(A)=\Fq^k]}.
$$
Separating the $c=0$ term, which equals $1$, and subtracting it from both sides of the above equation gives
$$
\sum_{c\ge1} (-1)^{c}q^{\binom{c}{2}} \qbin{k-a}{c} = \bone_{[\spn(A)=\Fq^k]}-1 = -\bone_{[\spn(A)\neq\Fq^k]}.
$$
Multiplying both sides by $-1$ and using \eqref{eq:qie1} proves the claim.
\end{IEEEproof}

Using Lemma \ref{lem:qie}, we now present the exact formula of $\E[\tau(\mathbf{G},\mathcal{P})]$ 

\begin{theorem}\label{thm:distributed-exact-representations}
Let $V_t$ be the space defined in \eqref{eq: column space by round t}, i.e., the space spanned by the columns drawn by round $t$.  Let $\{0\}=\mathcal{V}_0 \lneq \mathcal{V}_1 \lneq \cdots \lneq \mathcal{V}_L=\mathbb F_q^k$ be the chain of subspaces obtained from $(V_t)_{0\le t\le \tau(\mathbf{G},\mathcal P)}$ by deleting repeated consecutive subspaces. Then, $L\le k$ a.s. and
\begin{equation}\label{eq:epoch-representation}
    \mathbb E[\tau(\mathbf{G},\mathcal P)]
    =
    \mathbb E\left[
        \sum_{\ell=0}^{L-1}
        \frac{1}{1-\rho_{\mathcal V_\ell}}
    \right].
\end{equation}
Moreover, for every integer $t\ge 0$,
\begin{equation}
\label{eq:mobius-tail-formula}
    \Pr\bigl(\tau(\mathbf{G},\mathcal P)>t\bigr)
    =
    \sum_{U\lneq\mathbb F_q^k}
    (-1)^{\mathrm{codim}(U)-1}
    q^{\binom{\mathrm{codim}(U)}{2}}
    \rho_U^t,
\end{equation}
where we set $x^0=1$ when $x=0$. Consequently,
\begin{equation}
\label{eq:mobius-expectation-formula}
    \mathbb E[\tau(\mathbf{G},\mathcal P)]
    =
    \sum_{U\lneq\mathbb F_q^k}
    (-1)^{\mathrm{codim}(U)-1} q^{\binom{\mathrm{codim}(U)}{2}}
    \frac{1}{1-\rho_U}.
\end{equation}
\end{theorem}

\begin{IEEEproof}
We first note that $\rho_U<1$ for every $U\lneq\mathbb F_q^k$. Indeed, if $\rho_U=1$, then every factor $m_i(U)/n_i$ must equal one. This means every column stored in every container lies in $U$, contradicting $\mathrm{rank}(\mathbf{G})=k$.

Let $\mathcal F_t$ be the $\sigma$-algebra generated by the samples collected in the first $t$ rounds. Set $\sigma_0\triangleq 0$, and, as long as $V_{\sigma_\ell}\ne\mathbb F_q^k$, define $\sigma_{\ell+1}\triangleq \min\{t>\sigma_\ell: V_t\supsetneq V_{\sigma_\ell}\}$. Then $\mathcal{V}_\ell=V_{\sigma_\ell}$, and for $0\le \ell<L$, the holding time at $\mathcal V_\ell$ is denoted as $H_\ell\triangleq \sigma_{\ell+1}-\sigma_\ell$. Since $\rho_{U}<1$, the sampled columns will leave $U$ in finite time a.s. Hence, the chain of subspaces will reach $\mathbb{F}_q^k$ in finite time a.s. Moreover, since the span increases strictly between consecutive subspaces in the chain, we have $L\le k$ a.s.

Fix $0\le \ell<L$ and condition on the history up to time $\sigma_\ell$. If $\mathcal V_\ell=U$, then the span remains equal to $U$ after the next round if and only if all $M$ newly sampled columns lie in $U$. This event has probability $\rho_U$. Since different rounds are independent, we have $\Pr\bigl(H_\ell=h\mid \mathcal F_{\sigma_\ell}\bigr) =\rho_U^{h-1}(1-\rho_U)$ for integer $h\geq 1$. Hence, we have
$$
\mathbb E\bigl[ H_\ell \mid \mathcal F_{\sigma_\ell} \bigr] = \frac{1}{1-\rho_{\mathcal V_\ell}}.
$$
Because $\tau(\mathbf{G},\mathcal P)=\sigma_L=\sum_{\ell=0}^{L-1}H_\ell$ and $L\le k$, the tower property yields
\[
    \mathbb E[\tau(\mathbf{G},\mathcal P)]
    =
    \mathbb E\left[
        \sum_{\ell=0}^{L-1}
        \frac{1}{1-\rho_{\mathcal V_\ell}}
    \right],
\]
which proves \eqref{eq:epoch-representation}.


For $t\ge 0$, let $A_t\triangleq \{\mathbf{g}_j:~j\in J_t\}$ and $A_0\triangleq \varnothing$, where $J_t$ is defined by \eqref{eq: column index set by round t}. Then, we have $\tau(\mathbf{G},\mathcal P)>t$ if and only if $\spn(A_t)\ne\mathbb F_q^k$. Applying Lemma~\ref{lem:qie} to each realization of $A_t$
gives
$$
\mathbf 1\{\spn(A_t)\ne\mathbb F_q^k \} = \sum_{U\lneq\mathbb F_q^k} (-1)^{\mathrm{codim}(U)-1}q^{\binom{\mathrm{codim}(U)}{2}} \mathbf 1\{A_t\subseteq U\}.
$$
For a fixed subspace $U$, all samples collected during one round lie in $U$ with probability $\rho_U$. Independence across rounds therefore implies $\Pr(A_t\subseteq U)=\rho_U^t$. Taking expectations in the preceding indicator identity proves
\eqref{eq:mobius-tail-formula}.

Finally, using the tail-sum formula and the fact that the subspace
lattice is finite, we obtain
\begin{align*}
    \mathbb E[\tau(\mathbf{G},\mathcal P)]
    &=
    \sum_{t=0}^{\infty}
    \Pr\bigl(\tau(\mathbf{G},\mathcal P)>t\bigr)\\
    &=
    \sum_{U\lneq\mathbb F_q^k}
    (-1)^{\mathrm{codim}(U)-1}
    q^{\binom{\mathrm{codim}(U)}{2}}
    \sum_{t=0}^{\infty}\rho_U^t\\
    &=
    \sum_{U\lneq\mathbb F_q^k}
    (-1)^{\mathrm{codim}(U)-1}
    q^{\binom{\mathrm{codim}(U)}{2}}
    \frac{1}{1-\rho_U}.
\end{align*}
This proves
\eqref{eq:mobius-expectation-formula}.
\end{IEEEproof}

For the single-container case, where $M=1$, by definition, we have $\pi(U)=\frac{|\mathbf{G}\cap U|}{n}=\Pr(A_1\subseteq U)$. Note that the dimension of the span of the drawn vectors can increase only by one in each round. Consequently, we have $L=k$ in \eqref{eq:epoch-representation}, which leads to the following immediate corollary.

\begin{corollary}\label{cor:single-container-stage-decomposition}
For $0\le i\le k$, let $V^{(i)}$ be the span of the sampled columns at the first time its dimension becomes $i$. Then
\begin{equation}
\label{eq:single-container-stage-formula}
    \E[\mathcal{C}] = \sum_{i=0}^{k-1}
    \mathbb E\left[
        \frac{1}{1-\pi(V^{(i)})}
    \right] =
    \sum_{U\lneq\mathbb F_q^k}
    (-1)^{\mathrm{codim}(U)-1}q^{\binom{\mathrm{codim}(U)}{2}}
    \frac{1}{1-\pi(U)}.
\end{equation}
\end{corollary}

\subsection{The Optimality of MDS Codes in Distributed Model}

In~\cite{BSGY25}, Bar-Lev \textit{et al.} studied the coverage depth problem for general codes and proved that MDS codes are optimal in terms of minimizing the expected number of reads among all codes with the same parameters. Next, we show that MDS codes are also optimal in terms of the expected number of recovery rounds in the distributed storage model.

\begin{lemma}\label{lem:rank-function-domination}
Let $\mathbf{G}=(\mathbf{g}_1,\ldots,\mathbf{g}_n),~\widetilde{\mathbf{G}}=(\widetilde {\mathbf{g}}_1,\ldots,\widetilde{\mathbf{g}}_n) \in\mathbb F_q^{k\times n}$ be two different rank-$k$ generator matrices, and let $\mathcal P=\{P_1,\ldots,P_M\}$ be any fixed $M$-partition of $[n]$. Suppose that
\begin{equation}
\label{eq:rank-domination-assumption}
    \mathrm{rank}(\mathbf{G}|_J) \ge \mathrm{rank}(\widetilde{\mathbf{G}}|_J),~\forall~J\subseteq[n].
\end{equation}
Then $\tau(\mathbf{G},\mathcal{P}) \le \tau(\widetilde{\mathbf{G}},\mathcal{P})$ holds a.s. In particular, for every $t\ge 0$, we have $\Pr\bigl(\tau(\mathbf{G},\mathcal P)>t\bigr) \le \Pr\bigl(\tau(\widetilde{\mathbf{G}},\mathcal P)>t\bigr)$, and
$$
\mathbb E[\tau(\mathbf{G},\mathcal P)] \le \mathbb E[\tau(\widetilde{\mathbf{G}},\mathcal P)].
$$
\end{lemma}

\begin{IEEEproof}
Let $J_t\subseteq [n]$ be the set of distinct coordinate indices observed by the end of round $t$. Under the common coupling, the two codes have observed the same set $J_t$. Then, by assumption \eqref{eq:rank-domination-assumption}, we have
$$
\mathrm{rank}(\mathbf{G}|_{J_t}) \ge \mathrm{rank}(\widetilde{\mathbf{G}}|_{J_t})
$$
for every $t$. Thus, whenever $\widetilde{\mathbf{G}}|_{J_t}$ has rank $k$, the matrix $\mathbf{G}|_{J_t}$ also has rank $k$. Thus, $\tau(\mathbf{G},\mathcal P) \le \tau(\widetilde{\mathbf{G}},\mathcal P)$ holds for every realization of the sampled indices. The tail probability and expectation inequalities follow immediately.
\end{IEEEproof}

\begin{theorem}
\label{thm:distributed-mds-optimality}
Assume that an $[n,k]_q$ MDS code exists, and let $\mathbf{G}_{\mathrm{MDS}}$ be one of its generator matrices. Then, for every rank-$k$ generator matrix $\mathbf{G}\in\mathbb F_q^{k\times n}$ and every fixed $M$-partition $\mathcal P$ of $[n]$, it holds that
$$
\mathbb E[\tau(\mathbf{G}_{\mathrm{MDS}},\mathcal P)] \le \mathbb E[\tau(\mathbf{G},\mathcal P)].
$$
\end{theorem}

\begin{IEEEproof}
For every $J\subseteq[n]$, the MDS property yields that $\mathrm{rank}\bigl(\mathbf{G}_{\mathrm{MDS}}|_J\bigr) = \min\{|J|,k\} \ge \mathrm{rank}(\mathbf{G}|_{J})$. Hence, by Lemma~\ref{lem:rank-function-domination},  $\tau(\mathbf{G}_{\mathrm{MDS}},\mathcal P) \le \tau(\mathbf{G},\mathcal P)$ holds a.s. Moreover, we also have
$$
\mathbb E[\tau(\mathbf{G}_{\mathrm{MDS}},\mathcal P)] \le \mathbb E[\tau(\mathbf{G},\mathcal P)].
$$
This completes the proof.
\end{IEEEproof}

\begin{remark}\label{ex:nonmds}
We note that MDS codes are not the only optimal codes in terms of distributed coverage depth. For example, let $M=k$ and $n=ak$ for an integer $a\ge 2$. For each $j\in [M]$, container $j$ stores $a$ nonzero scalar multiples of the $j$-th standard basis vector: 
$$
\mathbf{G}_{j}=(\alpha_{j,1}\mathbf{e}_j,\alpha_{j,2}\mathbf{e}_j,\ldots,\alpha_{j,a}\mathbf{e}_j),
$$
where $\alpha_{j,\ell}\in \mathbb{F}_q^{*}$ for each $\ell\in [a]$. This is the generator matrix of a repetition code with repetition level $a$. Clearly, it is not an MDS code and its generator matrix has rank $k$. Note that the sampled columns contain a nonzero multiple of each $\mathbf{e}_j$ for $j\in[k]$ after the first round. Hence, $\tau(\mathbf{G},\mathcal{P})=1$, which is optimal.
\end{remark}

\subsection{General Lower Bounds on $\E[\tau(\mathbf{G},\mathcal{P})]$ and One-Round-Perfect Layouts}

In Section~\ref{subsec: exact formula}, we provided an exact formula for $\E[\tau(\bbG,\cP)]$  under a general layout $(\bbG,\cP)$. In this section, we first establish a lower bound on  $\E[\tau(\bbG,\cP)]$, and hence, on the expected total number of reads $M\E[\tau(\bbG,\cP)]$.

\begin{theorem}\label{thm:floor}
Let $M$ be an integer satisfying $1\leq M\leq k$, and let $r\triangleq \left\lceil \frac{k}{M}\right\rceil$. Then, for any rank-$k$ generator matrix $\mathbf{G}\in\mathbb{F}_q^{k\times n}$ and any $M$-partition $\mathcal{P}$, we have
$$
\mathbb{E}\!\left[\tau(\mathbf{G},\mathcal{P})\right] \ge r + \frac{\bigl(\frac{q^{k-M}-1}{q^k-1}\bigr)^r}{1 - \frac{q^{k-M}-1}{q^k-1}}.
$$
Moreover, when $M=k$ the lower bound is attained by the staircase partition. That is, \emph{every} $[n,k]_q$ linear code admits a one-round-perfect layout.
\end{theorem}

\begin{IEEEproof}
After $t$ rounds, at most $Mt$ columns (counted with multiplicity) have been observed. 
If $t<r$, then $Mt\le M(r-1)<k$, hence the observed columns cannot span $\mathbb{F}_q^k$. 
Consequently, $\Pr(\tau(\mathbf{G},\mathcal{P})>t)=1$ for $0\le t<r$. 
Using the tail-sum formula for a nonnegative integer random variable,
\begin{align}
\mathbb{E}\!\left[\tau(\mathbf{G},\mathcal{P})\right] = \sum_{t=0}^{\infty} \Pr(\tau(\mathbf{G},\mathcal{P})>t) 
\ge r + \sum_{t=r}^{\infty} \Pr(\tau(\mathbf{G},\mathcal{P})>t). \label{eq: tail-sum estimate}
\end{align}

For any $(k-1)$-dimensional hyperplane $H$ and any $t\ge r$, the event that all $M$ samples in each of the first $t$ rounds fall inside $H$ implies $\tau(\mathbf{G},\mathcal{P})>t$. Hence, $\Pr(\tau(\mathbf{G},\mathcal{P})>t) \ge \rho_H^t$, and summing over $t\ge r$ yields
$$
\sum_{t=r}^{\infty} \Pr(\tau(\mathbf{G},\mathcal{P})>t) \ge \frac{\rho_H^r}{1-\rho_H}.
$$

Let $\mathcal{P}=\{P_1,\ldots,P_M\}$ be a partition of $[n]$ with $|P_i|=n_i$. Denote by $\mathcal{H}$ the set of all $(k-1)$-dimensional subspaces of $\mathbb{F}_q^k$, with $|\mathcal{H}| = \frac{q^k-1}{q-1}$. Then
\begin{align*}
\sum_{H\in\mathcal{H}} \rho_H 
&= \frac{1}{\prod_{i=1}^M n_i} \sum_{H\in\mathcal{H}} 
   \sum_{x_1\in P_1}\!\cdots\!\sum_{x_M\in P_M} \mathbf{1}_{\{x_1,\dots,x_M\in H\}} \\
&= \frac{1}{\prod_{i=1}^M n_i} \sum_{x_1\in P_1}\!\cdots\!\sum_{x_M\in P_M} 
   \bigl|\{H\in\mathcal{H} : \langle x_1,\dots,x_M\rangle \subseteq H\}\bigr|\\
&\ge \frac{1}{\prod_{i=1}^M n_i} \Bigl(\prod_{i=1}^M n_i\Bigr) \frac{q^{k-M}-1}{q-1} = \frac{q^{k-M}-1}{q-1},
\end{align*}
where the inequality follows from $\dim(\langle x_1,\dots,x_M\rangle)\leq M$. 
Therefore,
$$
\frac{1}{|\mathcal{H}|}\sum_{H\in\mathcal{H}} \rho_H \ge \frac{\frac{q^{k-M}-1}{q-1}}{\frac{q^k-1}{q-1}} = \frac{q^{k-M}-1}{q^k-1}.
$$
Consequently, there exists a hyperplane $H_0$ such that $\rho_{H_0} \ge \frac{q^{k-M}-1}{q^k-1}$. 
Since $x\mapsto x^r/(1-x)$ is increasing on $[0,1)$, we obtain
\begin{align}
\sum_{t=r}^{\infty} \Pr(\tau(\mathbf{G},\mathcal{P})>t) 
\ge \frac{\rho_{H_0}^r}{1-\rho_{H_0}} 
\ge \frac{\bigl(\frac{q^{k-M}-1}{q^k-1}\bigr)^r}{1 - \frac{q^{k-M}-1}{q^k-1}}. \label{eq: tail-sum estimate2}
\end{align}
Finally, the bound follows by \eqref{eq: tail-sum estimate} and \eqref{eq: tail-sum estimate2}.

Now, we consider the case where $M = k$. We first introduce the staircase partition. For any generator matrix $\mathbf{G}$, we define the partition blocks $P_1, \dots, P_k$ inductively as follows:
\begin{itemize}
\item For $i = 1$, we select an arbitrary column $\mathbf{g}_{j_1}$ from $\mathbf{G}$ and set $P_1 = \{j \in [n] \colon \mathbf{g}_j \in \langle \mathbf{g}_{j_1} \rangle\}$.
\item For $i>1$, select an arbitrary index $j_i\in [n]\setminus\bigcup_{\ell<i}P_\ell$, and set
\[ P_i=\mathset{j\in [n] ~:~ \mathbf{g}_j\in \langle \mathbf{g}_{j_1},\dots,\mathbf{g}_{j_i}\rangle} \setminus \bigcup_{\ell=1}^{i-1}P_\ell. \]
\end{itemize}
Write $V_i\triangleq \langle \bbg_{j_1},\dots,\bbg_{j_i}\rangle$. By construction, it holds that $\dim V_i=i$, parts $P_1,\dots,P_k$ are nonempty and pairwise disjoint, $\bigcup_{i=1}^{k}P_i=[n]$ by $\rank(\bbG)=k$, and $\bbg_j\in V_i\setminus V_{i-1}$ for every $j\in P_i$ (with $V_0=\mathset{0}$). Hence, if $h_i\in P_i$ for $i\in[k]$, then, since $\bbg_{h_i}\in V_i\setminus V_{i-1}=V_i\setminus \langle \bbg_{h_1},\dots,\bbg_{h_{i-1}} \rangle$, we have $\langle\bbg_{h_1}, \dots,\bbg_{h_i}\rangle =V_i$ for every $i\in[k]$. In particular, $\rank(\bbg_{h_1},\dots,\bbg_{h_k})=k$, so that $\tau(\bbG,\cP)=1$ a.s. and the lower bound is attained.
\end{IEEEproof} 

\begin{remark}\label{rem:bound-is-counting}
The additive term in Theorem~\ref{thm:floor} is small. Since $\rho\triangleq\frac{q^{k-M}-1}{q^{k}-1}\leq q^{-M}$ and $r=\ceilenv{k/M}$, one has $\rho^{r}/(1-\rho)=O(q^{-k})$. The content of the bound is therefore essentially the counting bound $\E[\tau]\geq \ceilenv{k/M}$ together with its integrality, which is what we exploit later in Theorem~\ref{thm:intclass}. When $\mathcal{C}$ is a simplex code and $M=1$, the bound only yields $\E[\cC]\geq k+O(q^{-k})$, whereas Theorem~\ref{thm: main} below yields a much stronger bound $\E[\cC]\geq k+O(q^{-1})$.
\end{remark}

By Theorem~\ref{thm:floor}, every $[n,k]_q$ linear code admits a one-round-perfect $k$-partition, namely the staircase partition. In this paper, we also establish another sufficient condition for a layout $(\mathbf{G},\mathcal{P})$ to be  one-round-perfect, along with several explicit constructions based on Hamming codes and first-order Reed--Muller codes. Due to space constraints, we defer the detailed proofs and descriptions of these results to Appendix~\ref{app:one-round-perfect additional}.

\begin{remark}\label{rem:unbalanced-one-round}
The sizes of different parts in the staircase construction vary substantially, hence, the resulting partition is highly unbalanced, as illustrated by the layouts for simplex codes in Lemma~\ref{lem: optimality of simplex code for multiple containers},  for Hamming codes in Corollary~\ref{cor:hammingfloor}, and for first-order Reed--Muller codes in Theorem~\ref{thm:rm}. Thus, achieving one-round perfection may come at the expense of severe size imbalance among containers. If equal container capacities are required, the full saving $\mathrm{Ov}(\mathcal{C})$ may no longer be attainable, yet distributed sampling can still yield a strict gain over the single-container model, as illustrated below.
\end{remark}


We first recall the definition of first-order Reed--Muller code. Let $q$ be a prime power and $m \ge 1$ be an integer. Let $\mathbf{p}_1, \mathbf{p}_2, \dots, \mathbf{p}_{q^m}$ denote all the $q^m$ points in $\mathbb{F}_q^m$. The \emph{$q$-ary first-order Reed--Muller code} of length $n = q^m$ and dimension $k=m+1$, denoted by $\mathrm{RM}_q(1, m)$, is defined as the code with generator matrix:
$$
\mathbf{G} = \begin{pmatrix}
1 & 1 & \cdots & 1 \\
\mathbf{p}_1 & \mathbf{p}_2 & \cdots & \mathbf{p}_{q^m}
\end{pmatrix}.
$$

\begin{example}
\label{ex:rm13-balanced}
Let $\mathcal C=\mathrm{RM}(1,3)$, whose standard generator matrix has columns $g_{\mathbf p}=(1,\mathbf{p}^{T})^{T}$, $\mathbf{p}\in\mathbb{F}_2^3$. A balanced $4$-partition of its eight columns is therefore a partition of the eight points of $\mathbb{F}_2^3$ into four pairs.

We first show that no such partition is one-round-perfect. Note that four columns $\mathbf g_{\mathbf p_1},\ldots,\mathbf g_{\mathbf p_4}$ fail to span $\mathbb F_2^4$ if and only if $\mathbf p_1,\ldots,\mathbf p_4$ are affinely dependent. Since these four points are distinct, this is equivalent to their forming an affine plane in $\mathbb{F}_2^3$. Note that there are $14$ affine planes in $\mathbb{F}_2^3$, and every pair of points is contained in exactly $3$ of them. Hence, for any balanced $4$-partition, at most $4\cdot 3=12$ planes contain an entire pair. Hence, at least $14-12=2$ affine planes contain no entire pair. Since each such plane contains four points and  intersects each of the four pairs in at most one point, it must contain exactly one point from each pair. The corresponding four columns are linearly dependent. Consequently, no balanced $4$-partition of columns of $\mathrm{RM}(1,3)$'s generator matrix is one-round-perfect.

More precisely, for a balanced partition $\mathcal P$, let $N(\mathcal P)$ denote the number of affine planes containing exactly one point from each pair. In the first round, the selected columns fail to span $\Fq^k$ if they lie in one of the $N(\mathcal P)$ such affine planes. Thus, conditioned on having selected one such plane, in each subsequent round the spanning of $\Fq^k$ is completed as soon as at least one container returns the other point of its stored pair, which occurs with probability $1-\frac{1}{2^4}=\frac{15}{16}$. Therefore,
$$
\mathbb E[\tau(\mathbf G,\mathcal P)] = 1+\frac{N(\mathcal P)}{16}\cdot\frac{16}{15} = 1+\frac{N(\mathcal P)}{15}.
$$

The preceding counting argument gives $N(\mathcal{P}) \ge 2$, and hence $\mathbb{E}[\tau(\mathbf{G},\mathcal{P})] \ge \frac{17}{15}$ for every balanced $4$-partition. The equality is attained, for example, by
\begin{align*}
\mathcal{P}^\star =
    & \bigl\{
      \{(1,0,0,0),(1,1,0,0)\},
      \{(1,0,1,0),(1,0,0,1)\}, \\
    & ~~\{(1,1,1,0),(1,0,1,1)\},
      \{(1,1,0,1),(1,1,1,1)\}
    \bigr\}.
\end{align*}
Moreover, the two affine planes containing exactly one point from each pair are $\{(0,0,0),(0,0,1),(1,1,0),(1,1,1)\}$ and $\{(1,0,0),(0,1,0),(0,1,1),(1,0,1)\}$. Thus, $N(\mathcal{P}^\star)=2$ and $\mathbb{E}[\tau(\mathbf{G},\mathcal{P}^\star)]=\frac{17}{15}$, which yields an expected total number of reads of $4\mathbb{E}[\tau(\mathbf{G},\mathcal{P}^\star)] = \frac{68}{15} \approx 4.533$.

For comparison, Theorem~\ref{thm:rm} (i), with $q=2$ and $m=3$, gives the single-container coverage depth $\mathbb E[\mathcal C] = \frac{115}{21} \approx 5.476$.
\end{example}

Note that $\tau(\mathbf{G},\mathcal{P}) \ge 1$ by definition. To better understand its extreme behavior, the next theorem characterizes precisely when $\tau(\mathbf{G},\mathcal{P}) = 1$.

\begin{theorem}\label{thm: oneround}
Let $\bbG\in \F_q^{k\times n}$ be a rank-$k$ generator matrix and let $\cP$ be an ordered $M$-partition of $[n]$. Then $\tau(\mathbf{G},\mathcal{P})=1$ a.s. if and only if, for every hyperplane $H$ of $\mathbb{F}_q^k$, there exists an $i\in[M]$ such that $\mathbf{G}_i\cap H=\emptyset$. Consequently, at least $k$ containers are necessary for $\tau(\mathbf{G},\mathcal{P})=1$ a.s.
\end{theorem}

\begin{IEEEproof}
Let $S$ denote the multiset of columns drawn in the first round, with one column drawn from each container. Then $\tau(\mathbf{G},\mathcal{P})=1$ a.s. if and only if $\spn(S)=\mathbb{F}_q^k$ a.s. Moreover, $S$ fails to span $\mathbb{F}_q^k$ if and only if it is contained in some hyperplane of
$\mathbb{F}_q^k$.

Suppose first that, for every hyperplane $H$ of $\mathbb{F}_q^k$, there exists an $i\in[M]$ such that no column of $\mathbf{G}_i$ lies in $H$. Then the first-round outcome can not be contained in a hyperplane: if all columns in $S$ were contained in some hyperplane $H$, then the column drawn from container $i$ would have to lie in $H$, contradicting the choice of $i$. Hence every first-round outcome spans $\mathbb{F}_q^k$, and therefore $\tau(\mathbf{G},\mathcal{P})=1$ a.s.

Conversely, suppose that there exists a hyperplane $H_0$ such that each container contains at least one column lying in $H_0$. That is, for each $i\in[M]$, there is an $\ell_i\in P_i$ such that $\mathbf{g}_{\ell_i}\in H_0$. Note that 
$$
\Pr(S=\{\mathbf{g}_{\ell_1},\ldots,\mathbf{g}_{\ell_M}\})=\prod_{i=1}^M n_i^{-1}>0.
$$
This implies that all drawn columns lie in $H_0$ has positive probability, and hence they do not span $\mathbb{F}_q^k$. Thus $\Pr(\tau(\mathbf{G},\mathcal{P})=1)<1$. This proves the equivalence.

Finally, note that $\dim(\spn(S))\leq |S|=M$. Hence, $\tau(\mathbf{G},\mathcal{P})=1$ can hold only if $M\ge k$.
\end{IEEEproof}

We conclude this section by summarizing several representative code layouts in Table~\ref{tab:one-round-perfect-layouts}, all of which are one-round-perfect. The detailed proofs of the one-round-perfections of MDS codes and simplex codes are deferred to Sections \ref{sec: MDS distributed coverage} and \ref{sec: simplex distributed coverage}.

\begin{table}[htbp!]
  \centering
  \caption{Representative one-round-perfect layouts.}
  \label{tab:one-round-perfect-layouts}
  \begin{tabular}{ccccc}
    \toprule
    Code family & Parameters & Partition & Container sizes & Balanced? \\
    \midrule
    \thead{MDS code\\(Theorem \ref{thm:mdssplit})} & \thead{Any $n,k,q$ admitting \\ an $[n,k]_q$ MDS code} 
      & \thead{Any ordered $k$-partition}
      & \thead{Any $n_1,\ldots,n_k$, \\where $n_i\ge1$ and $\sum_i n_i=n$}
      & \thead{Yes \\ (when $k\mid n$)}\\
    \hline
    \thead{$q$-ary Hamming code\\(Corollary \ref{cor:hammingfloor})} & \thead{$[n,n-r]$ for any\\integer $r\geq 2$ with $n=\bra{r}$} 
      & General-position partition
      & $(r+1,1,\ldots,1)$
      & No \\
    \hline
    \thead{$\mathrm{RM}_q(1,m)$\\(Theorem \ref{thm:rm})} & \thead{$[q^m,m+1]$ for any\\integer $m\geq 1$}
      & Affine staircase partition
      & $(q^{m}-q^{m-1},\ldots,q-1,1)$
      & No \\
    \hline
    \thead{$q$-ary Simplex code\\(Lemma \ref{lem: optimality of simplex code for multiple containers})} & \thead{$[\bra{k},k]$ for any\\integer $k\geq 1$}
      & Projective staircase partition
      & $(q^{k-1},q^{k-2},\ldots,q,1)$
      & No \\
    \bottomrule
  \end{tabular}
\end{table}


\section{Distributed Coverage Depth Problem of MDS Codes}\label{sec: MDS distributed coverage}

In this section, we present several results on the distributed coverage depth problem of $[n,k]$ MDS codes. As we shall see, for certain parameter regimes, the expected total number of reads required to recover all $k$ information strands in the distributed storage setting is strictly smaller than that in the single-container setting. Throughout this section, for convenience, we use $\mathbf{G}_{\mathrm{MDS}}$ to denote the generator matrix of the corresponding MDS code.


For $[n,k]$ MDS codes, the distributed coverage depth problem defined in Section \ref{subsec: problem statement} reduces to the following variant of the classic coupon collector's problem (CCP): The model consists of $M$ distinct containers, where the $i$-th container holds $n_i$ distinct labeled balls. These $n = \sum_{i=1}^{M} n_i$ labeled balls correspond directly to the $n$ columns of $\mathbf{G}_{\mathrm{MDS}}$. At each round, we independently sample one ball uniformly at random (with replacement) from each container. The termination time $\tau(\mathbf{G}_{\mathrm{MDS}}, \mathcal{P})$ turns to be the number of rounds until $k$ distinct balls are collected. By this observation, we have the following immediate result.

\begin{theorem}\label{thm:mdssplit}
Let $\mathcal{C}$ be an MDS code $[n,k]$ with $k\ge2$ and $n\ge k$, Then, \emph{every} partition of $[n]$ into $k$ nonempty parts is one-round-perfect, which yields the full saving 
$$
\mathrm{Ov}(\mathcal{C})=\E[\mathcal{C}]-k=\sum_{i=1}^{k-1}\frac{i}{n-i}.
$$
\end{theorem}
\begin{IEEEproof}
Note that for MDS codes, any $k$ codeword symbols can recover the $k$ information symbols. Thus, the result follows directly from the definition of a one-round-perfect partition. Moreover, the value of the saving follows from the fact that $\mathbb{E}[\mathcal{C}] = n(H_{n} - H_{n-k})$ by~\cite[Corollary 1]{BSGY25} (see also~\cite[Theorem 2.2]{BRY26}).
\end{IEEEproof}

\begin{remark}\label{rmk: improved regimes}
Consider an $[n,k]$ MDS code $\mathcal{C}$ with $k\ge2$ and $n\ge k$. Since $\frac{i}{n} \le \frac{i}{n - i} \le \frac{i}{n-k+1}$, we have $\frac{k(k-1)}{2n} \le \mathrm{Ov}(\mathcal{C}) \le \frac{k(k-1)}{2(n-k+1)}$. As the rate $k/n$ varies, this yields a three-regime picture:
\begin{itemize}
    \item When $n \ge \frac{(k+2)(k-1)}{2}$, then the rate is at most $2/k$, which is very low. Then $\mathrm{Ov}(\mathcal{C}) \le \frac{k(k-1)}{2((k+2)(k-1)/2-k+1)} = 1$, so the full saving is at most one.
    \item When $ck \ll n < \frac{(k+2)(k-1)}{2}$, the rate is larger than $\frac{2}{k}(1-o(1))$ but far from constant. 
    Then, one has $\mathrm{Ov}(\mathcal{C}) \ge \frac{k(k-1)}{2n} \ge 1$ whenever $n < \frac{k(k-1)}{2}$, implying that the saving is at least one and grows like $\frac{k^{2}}{2n}$ as the rate increases.
    \item When $n=ck$ for a fixed constant $c>1$, $\mathrm{Ov}(\mathcal{C})$ grows linearly with $k$, so a constant \emph{fraction} of all reads is saved. We refer to this regime the \emph{constant rate} regime.
\end{itemize}
\end{remark}

Theorem~\ref{thm:floor} demonstrates that when $M=k$, there always exist one-round-perfect layouts that achieve full saving. Thus, what remains open is whether an \emph{intermediate} number of containers $M$ (with $1 < M < k$) can still achieve saving compared to the single-container model. Next, we answer this question in the affirmative by analyzing the distributed coverage depth of MDS codes for $1 < M < k$.

Recall from~\eqref{eq: column index set by round t} that $J_t$ denotes the set of distinct column indices observed by the end of round $t$. Hence, by the MDS property, we have $\tau(\mathbf{G}_{\mathrm{MDS}}, \mathcal{P}) = \min\{t \ge 1 \colon |J_t| \ge k\}$. Moreover, by the inclusion-exclusion principle, we have
\begin{equation}\label{eq:distinctlaw}
\Pr(|J_t \cap P_i| = d) = \binom{n_i}{d} \sum_{\ell=0}^{d} (-1)^{\ell} \binom{d}{\ell} \left(\frac{d-\ell}{n_i}\right)^{t},
\end{equation}
where $J_t \cap P_i$ is the set of distinct column indices drawn from $P_i$ by the end of round $t$.

In the following, we show an example for small parameters to illustrate the advantages of distributed storage compared to the single-server model.
\begin{example}
    For positive integers $n$, $k$, and $M$ such that $M\mid n$, consider an $[n, k]$ MDS code with an ordered balanced $M$-partition $\mathcal{P}$. Suppose that all containers have the same sizes $m = n/M$. Then, by \eqref{eq:distinctlaw}, we have
    $$
    \Pr\left(|J_t| \ge k\right) = \frac{1}{m^{tM}}
    \sum_{\substack{0 \le j_i \le \min\{m, t\} \\ j_1 + \cdots + j_M \ge k}}
    \prod_{i=1}^M \binom{m}{j_i} j_i! \begin{Bmatrix} t \\ j_i \end{Bmatrix},
    $$
    where $\begin{Bmatrix} \cdot \\ \cdot \end{Bmatrix}$ denotes the Stirling numbers of the second kind.  

Figure~\ref{fig:main} below compares the two models for $n=240$, $k=200$. In Figure~\ref{fig:1} the horizontal axis is the \emph{total} number of reads, 
so that the two families of curves are directly comparable in sequencing cost. The multi-container curves are steeper because the reads are spread deterministically over the containers, which suppresses part of the coupon-collector tail, the effect grows as $M$ increases, equivalently as the container size $n/M$ decreases. Figure~\ref{fig:2} plots the difference $\Pr(|J_t|\geq k)-\Pr(\tau^{(1)}\leq Mt)$ against the normalized read count $Mt/n$. It is nonnegative throughout and is maximized around $Mt/n\approx 1.750$, i.e., precisely in the range of read budgets at which the decision to stop sequencing is actually made. Outside that window both models have essentially decoded, or essentially not. 

\begin{figure}[htbp] 
    \centering
    \begin{subfigure}[b]{0.45\textwidth}
        \centering
        \includegraphics[width=\textwidth]{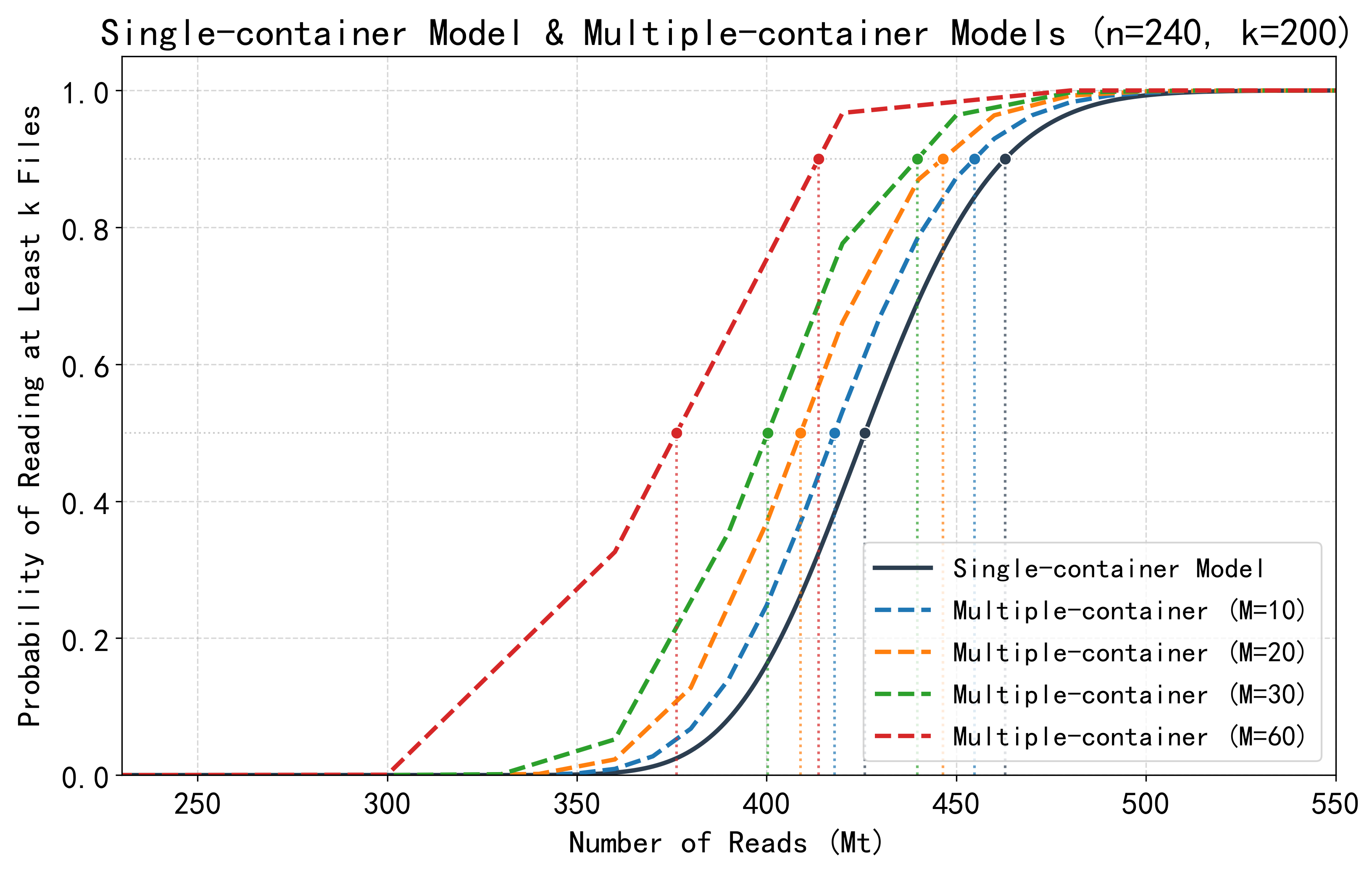}
        \caption{Probability of reading at least $k$ strands: single-container vs.\ multiple-container models.}
        \label{fig:1}
    \end{subfigure}
    \hfill
    \begin{subfigure}[b]{0.45\textwidth}
        \centering
        \includegraphics[width=\textwidth]{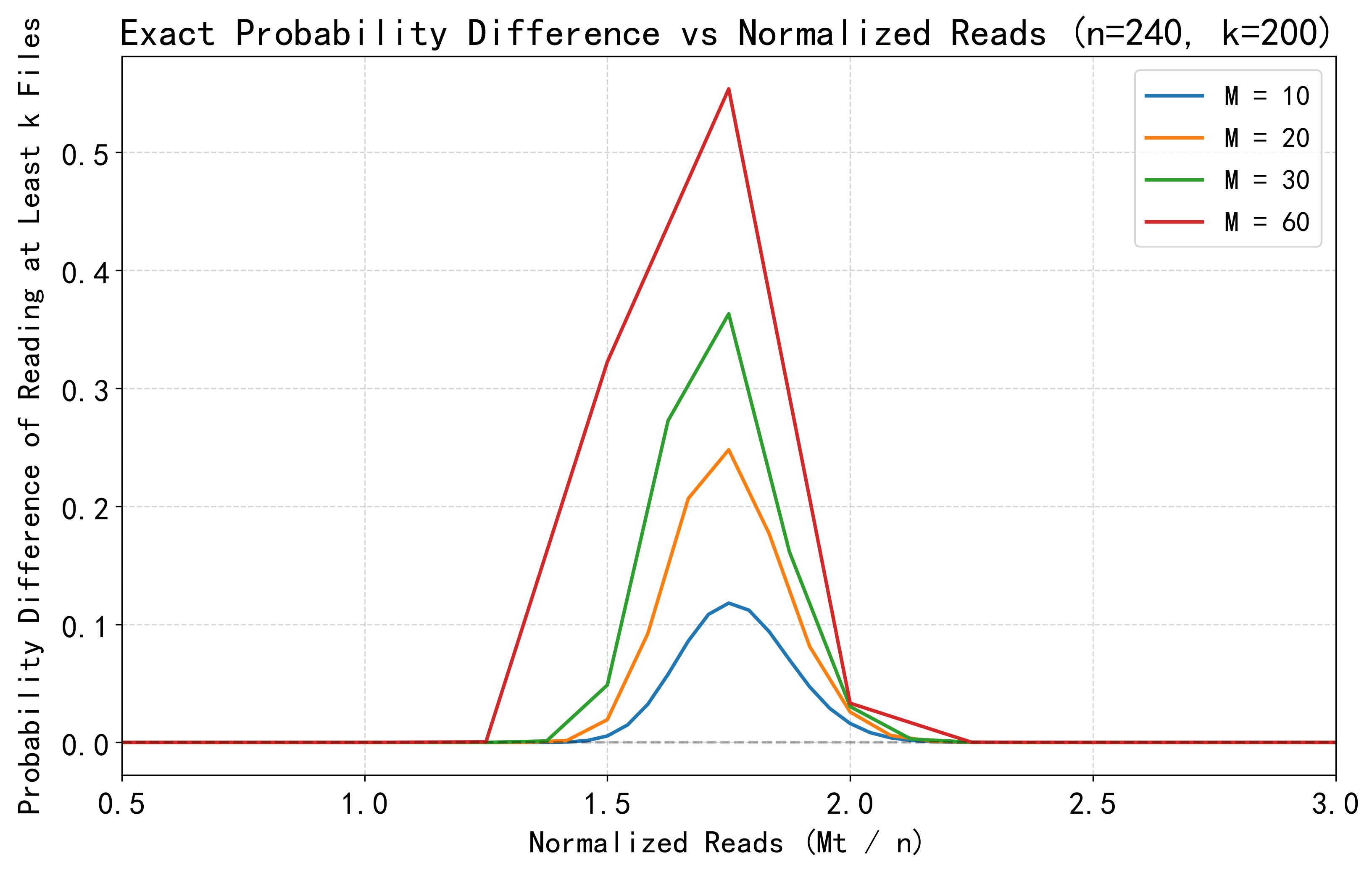}
        \caption{Probability difference of reading at least $k$ strands for the single-container and multiple-container models.}
        \label{fig:2}
    \end{subfigure}
    \caption{A comparison of the single container model and the multi-container model. Figure~\ref{fig:1} shows the probabilities of reading at least $k$ strands for the single-container and multiple-container models with $M \in \{10, 20, 30, 60\}$, and Figure~\ref{fig:2} displays the probability difference between the two models.}
    \label{fig:main}
\end{figure}
\end{example}

\subsection{Constant Container Size}

In this section, we consider the distributed coverage depth problem for MDS codes in the setting where $\mathcal{P}$ is an ordered balanced $M$-partition of $[n]$ (denoted by $\mathcal{P}_{B}$) and every container has constant size. In particular, we determine the asymptotic behavior of $\tau(\mathbf{G}_{\mathrm{MDS}}, \mathcal{P})$ and show that distribution indeed helps reduce the total number of reads.

Fix two integers $s\ge 2$ and $r\ge 2$, both independent of $k$, and denote $a_s\triangleq 1-\frac{1}{s}$. Let $\alpha>1$ be a constant satisfying
\begin{equation}
    \frac{1}{1-a_s^r} < \alpha < \frac{1}{1-a_s^{r-1}}.
    \label{eq:mds-alpha-range}
\end{equation}

\begin{theorem}\label{thm:mds-balanced-improvement}
Let $\{\mathcal{C}_k\}$ be a family of $[n = \alpha k, k]$ MDS codes with generator matrices $\{\mathbf{G}_k\}$, and $s\mid n$ for every $k$. Suppose the $n$ columns of $\mathbf{G}_k$ is partitioned arbitrarily into $M_k\triangleq \frac{n}{s}$ containers, each containing exactly $s$ columns. Let $\tau_k$ denote the first round by which $k$ different columns are observed. Then, as $k\rightarrow \infty$ (so $q\rightarrow \infty$),
$$
\E[\tau_k]=r+o(1),
$$
and hence the expected total number of drawings satisfies
$$
M_k\E[\tau_k] = \frac{\alpha r}{s}\,k+o(k).
$$
\end{theorem}

\begin{IEEEproof}
For each $k$ and $i\in[M_k]$ and $t\ge0$, let $D_{i,k}(t)$ denote the number of distinct columns observed from container $i$ during the first $t$ rounds, and define $D_k(t)\triangleq\sum_{i=1}^{M_k}D_{i,k}(t)$. Clearly, $D_k(t)$ is the total number of distinct columns observed during the first $t$ rounds, and we have $\tau_k\le t$ if and only if $D_k(t)\ge k$.

For each fixed $t$, the random variables $D_{1,k}(t),\ldots,D_{M_k,k}(t)$ are i.i.d. A fixed column in a container is not drawn during the first $t$ rounds with probability $a_s^t$. Hence, $\E[D_{i,k}(t)]=s(1-a_s^t)$. It follows that 
\begin{equation}
    \E[D_k(t)]
    =
    M_ks(1-a_s^t)
    =
    n(1-a_s^t).
    \label{eq:mds-occupancy-mean}
\end{equation}

Define $p_t\triangleq 1-a_s^t$. The first two inequalities in \eqref{eq:mds-alpha-range} are equivalent to
\begin{equation}
    \alpha p_{r-1}<1<\alpha p_r.
    \label{eq:mds-threshold-condition}
\end{equation}
Thus, $r$ is the first round in which the expected number of distinct observed
columns exceeds $k$ by a linear margin.

We next establish concentration around the mean. Since $0\le D_{i,k}(t)\le s$, by Hoeffding's inequality, for every $x>0$,
\begin{equation}
    \Pr\left(
        \left|D_k(t)-\E[D_k(t)]\right|\ge x
    \right)
    \le
    2\exp\left( -\frac{2x^2}{M_ks^2} \right).
    \label{eq:mds-hoeffding}
\end{equation}
Choose $\eta \triangleq \frac{1}{4} \min\left\{ 1-\alpha p_{r-1}, \alpha p_r-1 \right\}>0$. Since $n=\alpha k$, for $k$ sufficiently large, we have
$$
\begin{cases}
    \E[D_k(t)] \le (1-2\eta)k,~\text{when}~t\le r-1; \\
    \E[D_k(t)] \ge (1+2\eta)k,~\text{when}~t\ge r.
\end{cases}
$$
Combining these inequalities with \eqref{eq:mds-hoeffding} and using $M_k=\Theta(k)$, we  obtain,
$$
\begin{cases}
    \Pr(\tau_k\le t) = \Pr(D_k(t)\ge k) \le 2\exp\left( -\frac{2\eta^2k^2}{M_ks^2} \right) = e^{-\Omega(k)},~\text{when}~t\leq r-1;\\
    \Pr(\tau_k>t) = \Pr(D_k(t)<k) \le 2\exp\left(-\frac{2\eta^2k^2}{M_ks^2}\right) = e^{-\Omega(k)},~\text{when}~t\geq r.
\end{cases}
$$
Therefore, as $k\rightarrow \infty$, we have
\begin{equation}
    \Pr(\tau_k>t) \rightarrow
    \begin{cases}
        1, & 0\le t<r;\\
        0, & t\ge r,
    \end{cases}
    \label{eq:mds-tail-pointwise}
\end{equation}
for every fixed $t\geq 1$.

By \eqref{eq:mds-occupancy-mean}, the expectation number of columns that have not been observed by the end of round $t$ satisfies
$$
\E[n-D_k(t)]=na_s^t.
$$
If $\tau_k>t$, then fewer than $k$ distinct columns have been observed, and therefore $n-D_k(t)\ge n-k+1$. By Markov's inequality, this yields
\begin{equation}
    \Pr(\tau_k>t) \le \Pr(n-D_k(t)\ge n-k+1) \le \frac{na_s^t}{n-k+1}.
    \label{eq:mds-geometric-tail}
\end{equation}
Since $\frac{n}{k}=\alpha>1$, there exists a constant $C_\alpha>0$ such that $\frac{n}{n-k+1}\rightarrow C_\alpha$ as $k\rightarrow \infty$. Hence, we have
$$
\Pr(\tau_k>t) \le \min\{1,C_\alpha a_s^t\}.
$$
Then, using the tail-sum formula
and applying dominated convergence to \eqref{eq:mds-tail-pointwise}, we obtain
\begin{align*}
    \lim_{k\to\infty}\E[\tau_k] & =
    \lim_{k\to\infty} \sum_{t\ge0}\Pr(\tau_k>t)\\
    & = \sum_{t\ge0} \lim_{k\to\infty}\Pr(\tau_k>t)\\
    & = \sum_{t=0}^{r-1}1 = r,
\end{align*}
where the second equality follows since $\sum_{t\geq 0}C_{\alpha}a_s^t=\frac{C_{\alpha}}{1-a_s}<\infty$. Thus, we have $\E[\tau_k]=r+o(1)$. Moreover, since each complete round reads exactly $M_k$ columns, we further have
\begin{align}
    M_k\E[\tau_k] & =
    \frac{n}{s}\bigl(r+o(1)\bigr) = \frac{\alpha r}{s}\,k+o(k).
    \label{eq:mds-distributed-total}
\end{align}
This completes the proof.
\end{IEEEproof}

\begin{remark}\label{rmk: MDS multi-container vs single-container}
In addition to \eqref{eq:mds-alpha-range}, when $\alpha < (1 - e^{-r/s})^{-1}$, we have $\frac{\alpha-1}{\alpha} < e^{-r/s}$, which yields $\frac{r}{s} < \log\left(\frac{\alpha}{\alpha-1}\right)$. Combining $\E[\mathcal{C}_k] = n(H_n - H_{n-k})$ with \eqref{eq:mds-distributed-total}, we obtain in this case that
$$
\E[\mathcal{C}_k] - M_k\E[\tau_k] = \alpha \left( \log\left(\frac{\alpha}{\alpha-1}\right) - \frac{r}{s} \right) k + o(k).
$$
Hence we have $M_k\E[\tau_k]<\E[\mathcal{C}_k]$ for all sufficiently large $k$. That is, the expected total number of drawings of the multi-container model is strictly smaller than that of the single-server model. 

It is also worth noting that the number of containers satisfies $M_k = \Theta(k)$ and $M_k < k$. This holds because $s$ is fixed and satisfies $s>\alpha$: since $1 - a_s^{r-1} \ge 1 - a_s = \frac{1}{s}$, the upper bound on $\alpha$ in~\eqref{eq:mds-alpha-range} directly implies that $\alpha<s$.
\end{remark}

The saving established in Remark~\ref{rmk: MDS multi-container vs single-container} is a constant fraction $1-\frac{r/s}{\log\frac{\alpha}{\alpha-1}}$ of the total read cost. Since \eqref{eq:mds-alpha-range} forces $\alpha>(1-a_s^{r})^{-1}$ and hence $\log\frac{\alpha}{\alpha-1}<r\log\frac{s}{s-1}$, this fraction is at most
\[ 1-\parenv{s\log\frac{s}{s-1}}^{-1}
=\begin{cases}
    0.2787, & s=2;\\ 
    0.1779, & s=3;\\ 
    0.1310, & s=4, 
\end{cases}\]
and it tends to $0$ as $s\to\infty$. The regime is therefore genuinely a small-container regime: it requires $M_k=\alpha k/s=\Theta(k)$ containers, each holding a constant number $s$ of strands. 

\subsection{Polynomial Container Size}\label{sec: MDS large container size}

In this section, we consider the regime where each part of the balanced $M$-partition $\mathcal{P}_{B}$ of $[n]$ has size $n_i = n/M = k^{\varepsilon}$ for some fixed $0 < \varepsilon < 1$. Specifically, using the properties of negatively associated random variables, we derive the asymptotic behavior of $\mathbb{E}[\tau_k]$ as $k \to \infty$. The key properties regarding negative association, along with the detailed concentration and tail bounds, are provided in Appendix~\ref{app:mds-technical}.

Fix constants $\alpha > 1$ and $\varepsilon \in (0,1)$. For each integer $k\geq 1$, let $n = \alpha k$, $m_k = k^\varepsilon$, and $M_k = n/m_k$, and assume for simplicity that all these quantities are integers. Next, using Theorem~\ref{thm: CH bound for NA variables}, we prove the following result.

\begin{theorem}\label{thm: mds-balanced-nonimprovement}
Let $\mathcal{C}_k$ be a family of $[\alpha k, k]_q$ MDS codes and let $\tau_k$ be defined as in Theorem~\ref{thm:mds-balanced-improvement}. Set $\lambda \triangleq -\log \left(1 - \frac{1}{m_k}\right)$ and $t_c \triangleq \frac{\log(\alpha/(\alpha-1))}{\lambda}$. Then, as $k \to \infty$ (and thus $q \to \infty$), it holds that
\begin{equation}\label{eq:mds-nonimp-bound}
    \mathbb{E}\left[|\tau_k - t_c|\right] = O_\alpha\left(1 + \sqrt{\frac{m_k}{M_k}}\right).
\end{equation}
Consequently, we have
\begin{equation}\label{eq:mds-nonimp-exp}
    \mathbb{E}[\tau_k] = k^\varepsilon \log\left(\frac{\alpha}{\alpha-1}\right) + O_\alpha\left(1 + k^{\varepsilon-\frac{1}{2}}\right).
\end{equation}
\end{theorem}

\begin{IEEEproof}
We inherit the notations used in the proof of Theorem \ref{thm:mds-balanced-improvement}. Note that the stopping rule $D_k(t)\ge k$ is equivalent to the number of unseen coupons being at most $(\alpha-1) k$. Thus, in the following, we proceed the proof by estimating $\Pr(D_k(t)<k)$ and $\Pr(D_k(t)>k)$ according to different parameter regimes.

For the $j$-th column in container $i$, define the indicator variable
\begin{equation}\label{eq: (i,j)-indicator}
    I_{i,j}(t) \triangleq \mathbf{1}\{j\text{-th column in container } i \text{ has appeared by round } t\}.
\end{equation}
Then, by definition, we have $D_k(t)=\sum_{i=1}^{M_k}\sum_{j=1}^{m_k} I_{i,j}(t)$. Since the random variables $\{I_{i,j}(t)\}_{(i,j)\in [M_k]\times [m_k]}$ are negatively associated (see, e.g., \cite{BB25,Dubhashi_Panconesi_2009,DR98,JoagDevProschan83}) and independence across containers preserves this property \cite{DR98}, $D_k(t)$ is the sum of $n = M_k m_k$ negatively associated random variables. Recall that by \eqref{eq:mds-occupancy-mean}, we have $\mathbb{E}[D_k(t)] = n - n\left(1 - \frac{1}{m_k}\right)^t$. Thus, by Theorem~\ref{thm: CH bound for NA variables}, we obtain
\begin{equation}\label{eq: mds-balanced-nonimprovement}
    \Pr\left(|D_k(t) - \mathbb{E}[D_k(t)]| \ge x\right) \le 2\exp\left(-\frac{2x^2}{n}\right)
\end{equation}
for every $x > 0$. Furthermore, the definition of $t_c$ implies $\left(1 - \frac{1}{m_k}\right)^{t_c} = \frac{\alpha-1}{\alpha}$, which yields
$$
\mathbb{E}[n - D_k(\lceil t_c\rceil)] \leq n \left(1 - \frac{1}{m_k}\right)^{t_c} = (\alpha-1) k.
$$
Consequently, $\lceil t_c\rceil$ is the critical time when the expected number of unseen coupons reaches the threshold $(\alpha-1)k$.

We now estimate the tail of $\tau_k$. Let $r\geq 0$ be an integer and set $t=\floorenv{t_c+r}$. Since $\tau_k$ is integer valued, $\mathset{\tau_k>t_c+r}=\mathset{D_k(\floorenv{t_c+r})<k}=\mathset{D_k(t)<k}$. Since $t\geq t_c+r-1$ and $\E[D_k(\cdot)]$ is nondecreasing, we have
\begin{align}
    \mathbb{E}[D_k(t)]-k & \ge n-n\left(1-\frac{1}{m_k}\right)^{t_c+r-1}-k \nonumber \\
    &= (\alpha-1) k \left( 1-\left(1-\frac{1}{m_k}\right)^{r-1} \right). \label{eq: mds-balanced-nonimprovement1}
\end{align}
Note that for $0\le r\le m_k$, $1-\left(1-\frac{1}{m_k}\right)^r \ge c\frac{r}{m_k}$ for an absolute constant $0<c\leq 1$, and $M_k=\frac{n}{m_k}=\frac{\alpha k}{m_k}$. Hence, for $1\leq r\leq m_k+1$,
\[ 
\E[D_k(t)]-k\geq c(\alpha-1)k \frac{r-1}{m_k} =\frac{c(\alpha-1)}{\alpha} M_k(r-1),
\]
and Theorem \ref{thm: CH bound for NA variables} yields 
\begin{equation}\label{eq: mds-balanced-nonimprovement2}
\Pr(\tau_k>t_c+r)\leq \Pr\parenv{D_k(t)-\E[D_k(t)]\leq - \frac{c(\alpha-1)}{\alpha} M_k(r-1)}\leq 2\exp \parenv{-C_\alpha \frac{M_k(r-1)^2}{m_k}},
\end{equation} 
where $C_\alpha \triangleq 2\left(\frac{c(\alpha-1)}{\alpha}\right)^{2}$, and the bound is vacuous for $r=1$. 

Let $0\le r\le t_c-1$ and set
$t=\lfloor t_c-r\rfloor$. Since $\mathbb{E}[D_k(t)]$ is increasing in $t$,
\begin{align}
    k-\mathbb{E}[D_k(t)] & \ge k-n+n\left(1-\frac{1}{m_k}\right)^{t_c-r} \nonumber \\
    & = (\alpha-1) k\left(\left(1-\frac{1}{m_k}\right)^{-r}-1\right).  \label{eq: mds-balanced-nonimprovement3}
\end{align}
Since $\left(1-\frac{1}{m_k}\right)^{-r}=e^{\lambda r}$ and $\lambda=-\log(1-1/{m_k})\ge 1/{m_k}$, \eqref{eq: mds-balanced-nonimprovement3} can be simplified as
$$
k-\mathbb{E}[D_k(t)] \ge (\alpha-1) k(e^{\lambda r}-1) \ge (\alpha-1) k\lambda r \ge \frac{\alpha-1}{\alpha} M_k r.
$$
Hence,
\begin{align}
    \Pr(\tau_k<t_c-r) & \le \Pr(D_k(t)\ge k) \nonumber \\
    & \le \Pr\left(D_k(t) - \mathbb{E}[D_k(t)]\ge \frac{\alpha-1}{\alpha} M_k r\right) \nonumber \\
    &\le 2\exp\left(-2\left(\frac{\alpha-1}{\alpha}\right)^2\frac{M_k r^2}{m_k}\right). \label{eq: mds-balanced-nonimprovement4}
\end{align}

Combining the two tail estimates \eqref{eq: mds-balanced-nonimprovement2} and \eqref{eq: mds-balanced-nonimprovement4}, and writing $a^2\triangleq \frac{m_k}{M_k}$, we have, for $2\le r\le m_k$, $\Pr(|\tau_k-t_c|>r) \le 4\exp\left(-C_\alpha\frac{(r-1)^2}{a^2}\right)$.
Since $\tau_k$ is integer-valued, the standard tail-sum bound gives
\begin{align}
    \mathbb{E}[|\tau_k-t_c|] & \le 1+\sum_{\ell\ge 1}
    \Pr(|\tau_k-t_c|>\ell) \nonumber \\
    & = 1+ \sum_{\ell=1}^{m_k}\Pr(|\tau_k-t_c|>\ell)+\sum_{\ell>m_k}\Pr(|\tau_k-t_c|>\ell) \nonumber \\
    & \leq 2+ 4\sum_{\ell=2}^{m_k} \exp\left(-C_\alpha\frac{(\ell-1)^2}{a^2}\right)+\sum_{\ell>m_k}\Pr(|\tau_k-t_c|>\ell) \nonumber \\ 
    & \leq 2+ \tilde{C}_{\alpha}(1+a)+\sum_{\ell>m_k}\Pr(|\tau_k-t_c|>\ell), \label{eq: mds-balanced-nonimprovement6}
\end{align}
where 
$\tilde{C}_{\alpha}$ in \eqref{eq: mds-balanced-nonimprovement6} is a constant depends only on $\alpha$.

It remains to check that the contribution of $\sum_{\ell>m_k}\Pr(|\tau_k-t_c|>\ell)$ is negligible. 

Since $\lambda=m_{k}^{-1}(1+o(1))$ as $k\rightarrow \infty$, we have $t_c=O_\alpha(m_k)$. Hence,  there are $O_\alpha(m_k)$ number of values of $\ell$ which makes $\Pr(\tau_k<t_c-\ell)$ non-zero. Therefore, by \eqref{eq: mds-balanced-nonimprovement4}, we have
\begin{align}
    \sum_{\ell>m_k}\Pr(\tau_k-t_c<-\ell) & \leq O_\alpha(m_k)\cdot 2\exp\left(-2\left(\frac{\alpha-1}{\alpha}\right)^2\frac{M_k m_{k}^2}{m_k}\right) \nonumber \\
    & = O_\alpha(m_k)\cdot 2\exp\left(-2\left(\frac{\alpha-1}{\alpha}\right)^2 n\right) =o(1), \label{eq: mds-balanced-nonimprovement7}
\end{align}
as $k\rightarrow \infty$.

We next estimate $\sum_{\ell>m_k}\Pr(\tau_k-t_c>\ell)$. For a constant $A$ and $m_k<\ell\le A m_k\log k$, by \eqref{eq: mds-balanced-nonimprovement1}, we have
$$
\mathbb{E}[D(\lceil t_c+\ell\rceil)]-k \ge (\alpha-1) k \left( 1-\left(1-\frac{1}{m_k}\right)^{\ell} \right) \ge c_{\alpha} k,
$$
for some absolute constant $c_{\alpha}$ only depends on $\alpha$. Hence, by $\Pr(\tau_k-t_c>\ell)\leq \Pr(D(\lceil t_c+\ell\rceil)<k)$ and \eqref{eq: mds-balanced-nonimprovement}, we have
\begin{align}
    \sum_{\ell>m_k}^{Am_k\log k}\Pr(\tau_k-t_c>\ell) & \leq O(m_k\log k)\cdot 2\exp\left(-2{c_{\alpha}}^2\frac{k^2}{n}\right) =o(1), \label{eq: mds-balanced-nonimprovement8}
\end{align}
as $k\rightarrow \infty$. For $\ell>A m_k\log k$, since
$$
\mathbb{E}[n-D(\lceil t_c+\ell\rceil)] \le n \left(1-\frac{1}{m_k}\right)^{t_c+\ell} = (\alpha-1) k \left(1-\frac{1}{m_k}\right)^\ell,
$$
by Markov's inequality, this yields
$$
\Pr(\tau_k>t_c+\ell) = \Pr (n-D(\lceil t_c+\ell\rceil)> (\alpha-1) k) \le \left(1-\frac{1}{m_k}\right)^\ell.
$$
Hence, using the inequality $(1-x)^a \le \exp(-ax)$ for all $a \ge 0$ and $0 \le x \le 1$, we obtain
\begin{equation}\label{eq: mds-balanced-nonimprovement9}
    \sum_{\ell>A m_k\log k} \Pr(T_k>t_c+\ell) \le \sum_{\ell>A m_k\log k} \left(1-\frac{1}{m_k}\right)^\ell \le \frac{\left(1-\frac{1}{m_k}\right)^{A m_k\log k}}{1-\left(1-\frac{1}{m_k}\right)} \le m_k k^{-A}.
\end{equation}
Choosing $A=2$, by $\varepsilon\in(0,1)$ and $m_k=k^{\varepsilon}$, we have this becomes $o(1)$ uniformly as $k \rightarrow \infty$.

Finally, combining \eqref{eq: mds-balanced-nonimprovement6}, \eqref{eq: mds-balanced-nonimprovement7}, \eqref{eq: mds-balanced-nonimprovement8}, and \eqref{eq: mds-balanced-nonimprovement9} together, we have
$$
\mathbb{E}[|\tau_k-t_c|] \leq 1+ \tilde{C}_{\alpha}(1+a)+o(1)=O_{\alpha}\left(1+\sqrt{\frac{m_k}{M_k}}\right).
$$
Since $m_k=k^\varepsilon$ and $M_k=\alpha k^{1-\varepsilon}$, we have $\sqrt{\frac{m_k}{M_k}} = \alpha^{-\frac{1}{2}}k^{\varepsilon-\frac{1}{2}}$ and 
$$
t_c = \frac{\log(\frac{\alpha}{\alpha-1})}{-\log(1-k^{-\varepsilon})}=k^\varepsilon\log\left(\frac{\alpha}{\alpha-1}\right) + O_\alpha(1)
$$
as $k\rightarrow \infty$. This finally yields
$$
\mathbb{E}[\tau_k] = k^\varepsilon \log\left(\frac{\alpha}{\alpha-1}\right) + O_{\alpha}\left(1+k^{\varepsilon-\frac{1}{2}}\right),
$$
which completes the proof.
\end{IEEEproof}

\begin{remark}\label{rmk: MDS multi-container vs single-container0}
By \cite[Corollary 1]{BSGY25}, we have $\E[\mathcal{C}] = n\bigl(H_n-H_{n-k}\bigr) = \alpha\log\left(\frac{\alpha}{\alpha-1}\right)k+o(k)$. On the other hand, by Theorem~\ref{thm: mds-balanced-nonimprovement}, we have $\mathbb{E}[\tau_k]=k^\varepsilon\log\left(\frac{\alpha}{\alpha-1}\right)+O_{\alpha}\left(1+k^{\varepsilon-\frac{1}{2}}\right)$. Hence, the expected total number of reads across all $M$ containers is $M\mathbb{E}[\tau_k]=\alpha\log\left(\frac{\alpha}{\alpha-1}\right)k+O_{\alpha}\left(k^{1-\varepsilon}+k^{\frac{1}{2}}\right)$. In total, we have $M\mathbb{E}[\tau_k]/\mathbb{E}[\mathcal{C}]\longrightarrow 1$ when $k\to\infty$. Hence, in cases where each container stores $k^{\varepsilon}$ strands, distributing the reads across multiple containers enables parallel access without incurring any asymptotic increase in the total number of reads.
\end{remark}

\subsection{Low-Rate Obstructions}

In Theorem~\ref{thm: mds-balanced-nonimprovement}, we established that the distributed model provides no asymptotic savings in terms of total number of reads when all containers are required to have equal size $k^{\varepsilon}$. Next, we present another negative result, showing that for $M=2$, the distributed model offers no savings in the low-rate regime.

\begin{theorem}\label{thm:tworate}
Let $\mathcal{C}$ be an $[n,k]$ MDS code with $k\ge3$. Then, $2\E[\tau_k]> \E[\mathcal{C}]$ holds for the following cases
$$
\begin{cases}
    k=3~\text{and}~n\geq 3;\\
    k\geq 4~\text{and}~n\geq k^2.
\end{cases}
$$
\end{theorem}

\begin{IEEEproof}
Let $\mathcal{P}$ be any $2$-partition of $[n]$. Then, since each round contributes at most two distinct columns, we have $\tau_k\ge\lceil k/2\rceil$. Moreover, by Theorem~\ref{thm:mdssplit} and $n-i\ge n-k+1$, we have
\begin{equation}\label{eq:wastebound}
\E[\mathcal{C}]-k=\sum_{i=1}^{k-1}\frac{i}{n-i}\ \le\ \frac{\binom k2}{n-k+1} .
\end{equation}

We first consider the case where $k=3$. In this case, we have $\tau_k\geq \lceil3/2\rceil$, which implies that the total number of reads is at least $2\lceil3/2\rceil=4$. 
\begin{itemize}
    \item If $n\ge5$, since $\tfrac n{n-1}+\tfrac n{n-2}$ decreases in $n$, then
    $$
    \E[\mathcal{C}]=1+\tfrac n{n-1}+\tfrac n{n-2}\leq1+\frac{5}{4}+\frac{5}{3}=\tfrac{47}{12}<4.
    $$
    Hence, $2\E[\tau_k]> \E[\mathcal{C}]$ for $n\geq 5$.
    \item For $n=4$, we split the discussion into two cases: 1) both parts in $\mathcal{P}$ has size $2$; 2) one part $\mathcal{P}$ has size $1$. For the first case, we have $\tau_k>t$ if and only if drawings from both containers keep repeating their first column (i.e., $D_1=D_2=1$), so $\Pr(\tau_k>t)=2^{1-t}\cdot2^{1-t}=4^{1-t}$ and $2\E[\tau_k]=2\bigl(1+\sum_{t\ge1}4^{1-t}\bigr)=2\bigl(1+\tfrac43\bigr)=\tfrac{14}3$. 
    For the second case, we have $\tau_k\geq t$ if and only if the drawings from the container with $3$ columns keep repeating its first column. Thus, we have $\Pr(\tau>t)=3^{1-t}$ and $2\E[\tau]=2\bigl(1+\tfrac32\bigr)=5$. Both cases exceeding $\E[\mathcal{C}]=4(H_4-H_1)=\tfrac{13}3$.
    \item For $n=3$, the $[3,3]$-MDS code is just the identity code. In this case, the $2$-partition always contains one part of size $2$. Thus, we have $\Pr(\tau_k>t)=2^{1-t}$ and $2\E[\tau]=2(1+2)=6$, which is larger than $\E[\mathcal{C}]=3H_3=\tfrac{11}2$.
\end{itemize}

We now consider the case where $k\geq 4$.

When $k$ is odd. Then, we have $\tau_k\geq \lceil k/2\rceil$ and $2\E[\tau_k]\geq k+1$. However, $n\ge k^{2}$ gives
$n-k+1\ge k^{2}-k+1>\binom k2$. Hence, $2\E[\tau_k]\geq k+1>\E[\mathcal{C}]$ by \eqref{eq:wastebound}.

When $k=4$, the event $\tau_k>2$ occurs if and only if one of the columns drawn in the first two rounds repeats. Hence, 
$$
\Pr(\tau_k>2)=1-\Bigl(1-\frac{1}{n_1}\Bigr)\Bigl(1-\frac{1}{n_2}\Bigr)
=\frac{n_1+n_2-1}{n_1n_2}=\frac{n-1}{n_1n_2}\ \ge\ \frac{4(n-1)}{n^{2}},
$$
where the last inequality follows by $n_1n_2\le n^{2}/4$. Hence
$$
2\E[\tau_k]-4\ \ge\ 2\Pr(\tau>2)\ \ge\ \frac{8(n-1)}{n^{2}},
$$
and by \eqref{eq:wastebound}, $\E[\mathcal{C}]-4\ \le\ \frac{6}{n-3}$. Note that $8(n-1)(n-3)>6n^{2}$ holds for every $n\ge 16=k^{2}$. This yields that $2\E[\tau_k]>\E[\mathcal{C}]$.
 
When $k\ge6$ and $k$ is even. Let $r\triangleq k/2\ge3$. Then, $\tau_k\le r$ occurs iff the columns drawn from each container in the first $r$ rounds are distinct. Hence,
$$
\Pr(\tau>r)=1-\prod_{j=1}^{r-1}\Bigl(1-\frac{j}{n_1}\Bigr)\Bigl(1-\frac{j}{n_2}\Bigr)
\ \ge\ 1-\exp\Bigl(-\binom r2\Bigl(\frac{1}{n_1}+\frac{1}{n_2}\Bigr)\Bigr)
\ \ge\ 1-e^{-x},
$$
where $x=\frac{4\binom r2}{n}=\frac{k(k-2)}{2n}$ and the second inequality follows by $1-u\le e^{-u}$ and the last inequality follows by $\frac{1}{n_1}+\frac{1}{n_2}\ge\frac4n$. Then, by $n\ge k^{2}$ we have $x\le\frac{k-2}{2k}$, and hence $1-\frac x2\ge\frac{3k+2}{4k}$. Thus, by $1-e^{-x}\geq x-\frac{x^2}{2}$, we have
$$
2\E[\tau_k]-k\ \ge\ 2\Pr(\tau_k>r)\ \ge\ 2\Bigl(x-\frac{x^{2}}2\Bigr)\ \ge\ 2x\cdot\frac{3k+2}{4k}
=\frac{(k-2)(3k+2)}{4n} .
$$
On the other side, $n\ge k^{2}$ yields $n-k+1\ge n(1-\tfrac1k)$. Hence, by \eqref{eq:wastebound}, $\E[\mathcal{C}]-k\le\frac{k(k-1)}{2n(1-1/k)}=\frac{k^{2}}{2n}$. Note that
$(k-2)(3k+2)-2k^{2}=k^{2}-4k-4>0$ for $k\ge5$. Thus, $2\E[\tau_k]-k>\E[\mathcal{C}]-k$. This completes the proof.
\end{IEEEproof} 

For comparison, in Appendix~\ref{app:mds-technical}, we provide three exact finite-parameter examples of higher-rate MDS codes where a balanced two-container layout yields a strict saving.

\section{Distributed Coverage Depth Problem of Simplex Codes}\label{sec: simplex distributed coverage}

In this section, we consider the distributed coverage depth problem of simplex codes. We first study the general coverage depth problem in the single-container model. As a consequence, we show that the simplex code is optimal in terms of coverage depth in the single-container model, which resolves a conjecture proposed by Bertuzzo, Ravagnani, and Yaakobi~\cite{BRY26}. Next, we consider the distributed setting and show that the simplex code achieves optimal distributed coverage depth when the partition $\mathcal{P}$ has a staircase form. Furthermore, we obtain the asymptotic formula for the distributed coverage depth of the simplex code when the partition $\mathcal{P}$ is randomly balanced.

Throughout this section, we use $\mathbf{G}_{\mathrm{simp}}$ to denote the standard normalized generator matrix of simplex code, whose columns are the representatives of the projective points with first nonzero coordinate equal to $1$.

\subsection{Optimality of Simplex Codes in the Single-Container Model}

We start by recalling the coverage depth result of simplex codes in the single-container model.

\begin{theorem}[\!{\cite[Theorem 3.1]{BRY26}}]\label{thm: simplexE}
Let $\mathcal{C}\subseteq \mathbb{F}_q^{n}$ be the $q$-ary simplex code of dimension $k$ and length $n=\bra{k}$. Then we have
\begin{equation}\label{eq:simplexE}
\E[\mathcal{C}]=\sum_{i=0}^{k-1}\frac{q^{k}-1}{q^{k}-q^{i}} = (1-q^{-k})\sum_{j=1}^{k}\frac{1}{1-q^{-j}}.
\end{equation}
\end{theorem}

Throughout this section we use $\cC_{\mathrm{simp}}$ to represent the simplex code and write $\delta_q \triangleq \sum_{j\geq 1}\frac{1}{q^{j}-1}$, and
\begin{equation}\label{eq:simplexE2}
C^{(k)}_q \triangleq \mathrm{Ov}(\cC_{\mathrm{simp}}) =\E[\mathcal{C}]-k=\sum_{j=1}^{k}\frac{1-q^{j-k}}{q^{j}-1}.
\end{equation}
Clearly, $C^{(k)}_q<\delta_q$ for every $k$ and $C^{(k)}_q\to\delta_q$ as $k\to\infty$. 

For the single-container model, Corollary~\ref{cor:single-container-stage-decomposition} shows that the coverage depth of any code depends solely on the probability law $\pi$ of its column distribution across subspaces of $\mathbb{F}_q^{k}$. Hence, in this section, we write $\mathbb{E}[\mathcal{C}] = \mathbb{E}(\pi)$.

Recall that $\mathbb{P}$ denotes the set of projective points in $\mathbb{F}_q^k$. For a rank-$k$ generator matrix $\mathbf{G}=(\mathbf{g}_1,\dots,\mathbf{g}_n)$, we define
\begin{equation}\label{eq:pi-p-def}
\pi_\mathbf{p} \triangleq \frac{|\{j\in[n] \colon \langle\mathbf{g}_j\rangle = \mathbf{p}\}|}{n}
\end{equation}
for each $\mathbf{p}\in \mathbb{P}$. Then, for any subspace $U \le \mathbb{F}_q^k$, according to its definition in Section~\ref{subsec: exact formula}, we have $\pi(U) =\sum_{\mathbf{p}\in \mathbb{P}(U)} \pi_\mathbf{p}$, where $\mathbb{P}(U)$ denotes the set of projective points in $U$. With this notation, the simplex code, in which each projective point appears as a column exactly once, has a column law given by the uniform distribution on $\mathbb{P}$.

Next, we proceed with the proof of the optimality of simplex codes in the single-container model, starting with the following lemma.

\begin{lemma}\label{lem: mass}
For any probability law $\pi$ whose support spans $\mathbb{F}_q^{k}$, let $\bar{\pi}_i \triangleq \mathbb{E}[\pi(V^{(i)})]$ for each $0 \le i \le k$. Then, we have $\bar{\pi}_i \ge \frac{\gv{i}}{\gv{k}}$, for $0 \le i \le k$, where equality holds for every $i$ if and only if $\pi = u$.
\end{lemma}

\begin{IEEEproof}
We proceed the proof by induction on $i$. The case $i=0$ is immediate, since $V^{(0)}=\{\mathbf{0}\}$ and hence $\bar\pi_0=0={\gv 0}/{\gv k}$.

Now fix an integer $i$ with $1\le i\le k$, and assume that $\bar\pi_{i-1}\ge {\gv{i-1}}/{\gv k}$. Let $U$ be a fixed $(i-1)$-dimensional subspaces. The number of $i$-dimensional subspaces $W$ satisfying $U\leq W$ is
$$
N_i \triangleq \qbin{k-(i-1)}{1} = \gv{k-i+1},
$$
which is independent of $U$. Moreover, the sets $W\setminus U$, where $W$ ranges over all such $i$-dimensional subspaces, form a partition of $\mathbb{F}_q^k\setminus U$. For each such subspace $W$, define $a_W \triangleq \pi(W)-\pi(U)$. By the definition of $\pi(U)$, we have
$$
a_W = \sum_{\mathbf{p}\in \mathbb{P}(W)}\pi_{\mathbf{p}} -\sum_{\mathbf{p}\in \mathbb{P}(U)}\pi_{\mathbf{p}}
 = \sum_{\mathbf{p}\in \mathbb{P}(W)\setminus \mathbb{P}(U)} \pi_{\mathbf{p}} \ge 0.
$$
Hence,
$$
\sum_{\substack{
        U\leq W\leq \mathbb{F}_q^k \\ \dim(W)=i
    }}
    a_W = \pi(\mathbb{F}_q^k)-\pi(U) = 1-\pi(U).
$$

Conditioned on $V^{(i-1)}=U$ and the first draw leaving $U$, the resulting span is equal to $W$ with probability $a_W/(1-\pi(U))$. Therefore,
\begin{align}
    \E\!\left[\pi(V^{(i)})\mid V^{(i-1)}=U\right]
    &=
    \sum_{\substack{
        U\leq W\leq \mathbb{F}_q^k \\ \dim(W)=i
    }}
    \frac{a_W}{1-\pi(U)}\bigl(\pi(U)+a_W\bigr) \nonumber\\
    &=
    \pi(U)
    +
    \sum_{\substack{
        U\leq W\leq \mathbb{F}_q^k \\ \dim(W)=i
    }}
    \frac{a_W^2}{1-\pi(U)} .
    \label{eq:mass1}
\end{align}
By the Cauchy--Schwarz inequality, we obtain
$$
    \sum_{\substack{
        U\leq W\leq \mathbb{F}_q^k \\ \dim(W)=i
    }}
    a_W^2
    \ge
    \frac{1}{N_i}
    \Bigl(
    \sum_{\substack{
        U\leq W\leq \mathbb{F}_q^k \\ \dim(W)=i
    }}
    a_W
    \Bigr)^2
    =
    \frac{(1-\pi(U))^2}{N_i}.
$$
Substituting this bound into \eqref{eq:mass1} yields
\begin{align}
    \E\!\left[\pi(V^{(i)})\mid V^{(i-1)}=U\right] \ge
    \pi(U)+\frac{1-\pi(U)}{N_i} =
    \pi(U)\left(1-\frac{1}{N_i}\right)
    +\frac{1}{N_i}.
    \label{eq:mass2}
\end{align}

Taking expectations in both sides of \eqref{eq:mass2} over $U$, we obtain $\bar\pi_i \ge \bar\pi_{i-1}\left(1-\frac{1}{N_i}\right) + \frac{1}{N_i}$. Since $N_i=\gv{k-i+1}\ge 1$, the map
$x\mapsto x(1-\frac{1}{N_i})+\frac{1}{N_i}$ is nondecreasing. By the
induction hypothesis, this yields
$$
\bar\pi_i \ge \frac{\gv{i-1}}{\gv k}\left(1-\frac{1}{N_i}\right) + \frac{1}{N_i}.
$$
Then, using $\gv k-\gv{i-1} = \frac{q^k-q^{i-1}}{q-1} = q^{i-1}\gv{k-i+1} = q^{i-1}N_i$, we further obtain
\begin{align*}
    \frac{\gv{i-1}}{\gv k}\left(1-\frac{1}{N_i}\right)
    +\frac{1}{N_i}
    &=
    \frac{\gv{i-1}}{\gv k}
    +
    \frac{\gv k-\gv{i-1}}{\gv k\,N_i} \\
    &=
    \frac{\gv{i-1}}{\gv k}
    +
    \frac{q^{i-1}}{\gv k} =
    \frac{\gv i}{\gv k},
\end{align*}
where the last equality follows from $\gv i-\gv{i-1}=q^{i-1}$. This proves
the desired lower bound.


Finally, the equality in the Cauchy-Schwarz step at level $i$ forces the numbers $a_W$ to be constant over all $i$-dimensional subspaces $W$ containing $U$, for every possible $U$ in the support of $V^{(i-1)}$. Specifically, taking $i=1$ already suffices: for $i=1$, we have $V^{(0)}=\{0\}$, and the relevant subspaces $W$ are exactly the points $p\in\mathcal{P}$, with $a_W=\pi_p$. Equivalently, $\bar\pi_1=\sum_{p\in\mathcal{P}}\pi_p^{2}\geqslant\frac{1}{[k]_q}$ with equality if and only if $\pi_p\equiv1/[k]_q$, i.e., $\pi=u$. Conversely, if $\pi=u$, then $u(V^{(i)})=[i]_q/[k]_q$ for every $i$, so equality holds at every level.
\end{IEEEproof}

With the help of Lemma \ref{lem: mass}, we next show that simplex codes has the smallest coverage depth among all kinds of linear codes with the same code length and dimension, which confirms \cite[Conjecture 3.2]{BRY26}.

\begin{theorem}\label{thm: main}
Let $u$ be the uniform probability law on $\mathbb{P}$. For every probability law $\pi$ on $\mathbb{P}$, we have
$$
E(\pi)\ge E(u) = \sum_{i=0}^{k-1}\frac{q^k-1}{q^k-q^i}.
$$
Moreover, equality holds if and only if $\pi=u$. 
In particular, among length-$[k]_q$, dimension-$k$ linear codes, the $q$-ary simplex code is, up to isomorphism, the unique minimizer.
\end{theorem}

\begin{IEEEproof}
If $\pi(V^{(i)})=1$ with positive probability for some $i<k$, then Corollary \ref{cor:single-container-stage-decomposition} gives $E(\pi)=\infty$, and the desired inequality is trivial. Thus, in what follows, we assume that $\pi(V^{(i)})<1$ a.s. for every
$0\le i\le k-1$.

Let $\phi(x)=(1-x)^{-1}$, for $0\le x<1$. By Corollary \ref{cor:single-container-stage-decomposition} and Jensen's inequality, for each $0\le i\le k-1$, we have
$$   \E\!\left[\phi\!\left(\pi(V^{(i)})\right)\right] \ge \phi\!\left(\E[\pi(V^{(i)})]\right) = \phi(\bar\pi_i).
$$
Since $\phi$ is strictly increasing on $[0,1)$, Lemma~\ref{lem: mass} implies
\begin{align*}
    \E(\pi)
    &=
    \sum_{i=0}^{k-1}
    \E\!\left[\phi\!\left(\pi(V^{(i)})\right)\right] \ge
    \sum_{i=0}^{k-1}
    \phi(\bar\pi_i) \\
    & \ge
    \sum_{i=0}^{k-1}
    \phi\!\left(\frac{\gv i}{\gv k}\right) =
    \sum_{i=0}^{k-1}
    \frac{\gv k}{\gv k-\gv i} \\
    & =
    \sum_{i=0}^{k-1}
    \frac{q^k-1}{q^k-q^i}= \E(u),
\end{align*}
where the last equality follows by Theorem \ref{thm: simplexE}.

It remains to characterize the equality case. Since all the gaps in the two inequalities above are nonnegative, equality $E(\pi)=E(u)$ holds if and only if
equality holds term by term in both inequalities, namely,
$$   \E\!\left[\phi\!\left(\pi(V^{(i)})\right)\right] = \phi(\bar\pi_i)
$$
and $\phi(\bar\pi_i) = \phi\!\left(\frac{\gv i}{\gv k}\right)$ for all $0\le i\le k-1$. For the first equality, since $\phi(x)=(1-x)^{-1}$ is strictly convex on $[0,1)$, equality in Jensen's inequality holds if and only if
$$
\pi(V^{(i)})=\E[\pi(V^{(i)})]=\bar\pi_i.
$$
For the second equality, since $\phi$ is strictly increasing, it is equivalent to $\bar\pi_i={\gv i}/{\gv k}$ holds for all $0\le i\le k-1$. Thus, by Lemma~\ref{lem: mass}, the second equality holds if and only if $\pi=u$. Moreover, when $\pi=u$, we have $u(V^{(i)})={\gv i}/{\gv k}$ for every $i$, and hence the equality condition in Jensen's inequality is also satisfied. Therefore, equality in the theorem holds if and only if $\pi=u$.
\end{IEEEproof}

\begin{corollary}\label{rev:projective-stability}
Let $k\ge2$, let $N=[k]_q$, and let $\pi$ be a probability distribution on $\mathbb{P}$ whose support spans $\mathbb F_q^k$. Then
\begin{equation}\label{rev:stability-bound}
E(\pi)-E(u) \ge 
\frac{\|\pi-u\|_2^2}{(1-1/N)^2}.
\end{equation}
\end{corollary}

\begin{proof}
Let $d \triangleq \|\pi-u\|_2^2 = \sum_{p\in\mathbb P}\left(\pi_p - 1/N\right)^2$. By Lemma~\ref{lem: mass}, for each $0 \le i \le k-1$, the difference $\E\!\left[\phi\!\left(\pi(V^{(i)})\right)\right] - \frac{q^k-1}{q^k-q^i}$ is non-negative. Hence, considering only the term corresponding to $i=1$, we have
$$
\E(\pi) - \E(u) \ge \frac{1}{1-\bar\pi_1} - \frac{1}{1-1/N} = \frac{d}{(1-1/N-d)(1-1/N)},
$$
where the second equality follows from $\bar\pi_1 = \sum_{p\in\mathbb P}\pi_p^2 = 1/N + d$. Since the support of $\pi$ spans $\mathbb{F}_q^k$, it must contain at least two points. Hence, $\sum_p\pi_p^2 < 1$ and $0 < 1 - 1/N - d < 1 - 1/N$, which directly yields \eqref{rev:stability-bound}.
\end{proof}

\subsection{Full Saving for $M=k$ and No-Saving Obstructions for $M<k$}\label{sec: simplex codes multi-container}

As a direct consequence of Theorem~\ref{thm:floor}, the following lemma establishes the explicit construction of one-round-perfect partitions for the simplex code.

\begin{lemma}\label{lem: optimality of simplex code for multiple containers}
For each $1\le i\le k$, assign to container $i$ all columns of the form
$$
\mathbf{g} = (\underbrace{0,\ldots,0}_{i-1},1,*,\ldots,*)^{T},
$$
where the entries denoted by $*$ are arbitrary elements of $\mathbb{F}_q$, and denote the corresponding $k$-partition of $[(q^k-1)/(q-1)]$. Then, $\mathcal{P}$ is one-round-perfect, and  $\tau(\mathbf{G}_{\mathrm{simp}},\mathcal{P})=1$ a.s.
\end{lemma}

\begin{proof}
Let $W_k\subsetneq W_{k-1}\subsetneq\dots\subsetneq W_1=\F_q^{k}$ be a sequence of subspaces of $\F_q^k$, where $W_i=\langle \bbe_i,\bbe_{i+1},\dots,\bbe_k\rangle$ and $W_{k+1}\triangleq \{0\}$. Then, container $i$ is exactly the set of columns lying in $W_i\setminus W_{i+1}$, with $q^{k-i}$ elements, and any set of $k$ columns $(\bbh_1,\dots,\bbh_k)$ with $\bbh_i\in W_i\setminus W_{i+1}$ forms an upper triangular matrix with nonzero diagonal entries, hence a basis of $\F_q^{k}$. Hence, the result follows by the definition of one-round-perfection.
\end{proof}

\begin{remark}\label{rem:small-overhead}
By \eqref{eq:simplexE2}, $\mathrm{Ov}(\cC_{\mathrm{simp}})<\delta_q = \sum_{j\geq 1}(q^{j}-1)^{-1}$, with $\delta_2=1.6067$ and $\delta_q<1$ for $q\geq 3$. Thus, for simplex codes, whether distribution reduces the total number of reads concerns an overhead bounded by an absolute constant, independent of $k$ (at most about $1.61$ when $q=2$, and less than one when $q\geq 3$). 
\end{remark}

By Theorem \ref{thm:floor}, we have the following immediate corollary.

\begin{corollary}\label{cor: staircase-opt}
For any $q$-ary $k$-dimensional simplex code, we have $\min_{M,\mathcal{P}} M \mathbb{E}\left[\tau(\mathbf{G}_{\mathrm{simp}},\mathcal{P})\right] = k$.
This minimum is attained by the $k$-partition of $[(q^{k}-1)/(q-1)]$ given in Lemma~\ref{lem: optimality of simplex code for multiple containers}. Thus, distributing the simplex code into $k$ containers as in Lemma~\ref{lem: optimality of simplex code for multiple containers} achieves full savings $\mathrm{Ov}(\mathcal{C}_{\mathrm{simp}})$.
\end{corollary}

We stress that the weaker consequence $M\E[\tau]\geq k$ of Theorem~\ref{thm:floor} is not sufficient here, since the single-container coverage depth of the simplex code equals $k+C^{(k)}_q>k$. What is needed is the \emph{integrality} of the bound. The following theorem gives $M\E[\tau]\geq M\ceilenv{k/M}$, and by Remark~\ref{rem:small-overhead}, the overhead $C^{(k)}_q$ of the simplex code is bounded by an absolute constant, so this rounding alone already rules out any saving whenever $M\ceilenv{k/M}>k$.

\begin{theorem}\label{thm:intclass}
Let $\mathcal{C}$ be a $q$-ary simplex code of dimension $k$ with standard generator matrix $\mathbf{G}_{\mathrm{simp}}$, and let $M$ be any integer satisfying $1 < M < k$. Then, for every $M$-partition $\mathcal{P}$ of $[(q^k-1)/(q-1)]$, it holds that $M\mathbb{E}[\tau(\mathbf{G}_{\mathrm{simp}}, \mathcal{P})] \ge M\lceil k/M\rceil$. Moreover, we have the following two integrality obstructions:
\begin{enumerate}
\item[(i)] If $M \nmid k$ and $M \nmid (k+1)$, then $M\lceil k/M\rceil \ge k+2 > \mathbb{E}[\mathcal{C}]$.
\item[(ii)] If $M \mid (k+1)$ and either $q \ge 3$ or $k = 3$, then $M\lceil k/M\rceil = k+1 > \mathbb{E}[\mathcal{C}]$.
\end{enumerate}
\end{theorem}

\begin{IEEEproof}
The result $M\mathbb{E}[\tau(\mathbf{G}_{\mathrm{simp}}, \mathcal{P})] \ge M\lceil k/M\rceil$ follows directly by Theorem~\ref{thm:floor}. 

First, $M\nmid k$ gives
$M\lceil k/M\rceil=k+(M-(k\bmod M))\ge k+1$. Moreover, note that $M-(k\bmod M)=1$ if and only if $M\mid k{+}1$. Hence, for (i), we have $M\lceil k/M\rceil\geq k+2$, and for (ii), $M\lceil k/M\rceil= k+1$.

On the other hand, since $\frac{1-q^{j-k}}{q^{j}-1}<\frac{1}{q^{j}-1}$, we have $C_q^{(k)}<\delta_q \triangleq \sum_{j\geq 1}(q^j-1)^{-1}$. Since $q^{j}-1>q^{j-1}$ holds for $q\geq 2$ and $j\geq 3$, we have 
$$
\delta_{2}=1+\frac{1}{3}+\sum_{j\geq 3}(2^{j}-1)^{-1}<1+\frac{1}{3}+\sum_{j\geq 3}2^{1-j}=\frac{11}{6}<2,
$$
and for $q\geq 3$
$$
\delta_q\le \delta_3<\frac{1}{2}+\frac{1}{8}+\sum_{j\geq 3}3^{1-j}=\frac{1}{2}+\frac{1}{8}+\frac{1}{6}<1.
$$
Consequently, this proves (i) holds for every $q$ and (ii) for $q\ge3$. 

Finally, for $q=2$ and $k=3$, the results follows since $C_2^{(3)}=\E[\mathcal{C}]-3=\tfrac{47}{12}-3=\tfrac{11}{12}<1$.
\end{IEEEproof}

The two smallest parameter pairs not settled by the preceding integrality obstruction require a separate codimension-two argument. Corresponding arguments for $M=2$ is deferred to Appendix~\ref{app:simplex-exceptional}.

\subsection{Distributed Coverage Depth under Balanced Partitions}

Fix $q$, we next restrict our attention to the case where $M \mid \bra{k}$ and every part of the ordered $M$-partition $\mathcal{P}$ has the same size $\bra{k}/M$. 

 
\begin{remark}\label{rem:balanced-existence}
A balanced $M$-partition of the columns of $\bbG_{\mathrm{simp}}$ exists only if $M\mid [k]_q$, and this imposes a restriction on $k$. If prime $p$ divides both $M$ and $q$, then $[k]_q=1+q+\cdots+q^{k-1}\equiv 1 \bmod p$, so $M\nmid[k]_q$ for \emph{every} $k$. In particular, for $q=2$ no balanced $M$-partition exists for any even $M$. If $p\mid M$ and $p\nmid q$, then $p\mid[k]_q$ is equivalent to $d\mid k$, where $d=p$ when $p\mid q-1$ and $d=\mathrm{ord}_p(q)$ otherwise.\footnote{Here $\mathrm{ord}_p(q)$ denotes the order of $q$ modulo $p$, i.e., the smallest positive integer $d$ satisfying $q^d \equiv 1 \pmod{p}$.} 
For instance, for $q=3$ and $M=2$ only even $k$ are admissible. All results in this subsection are stated for admissible dimensions $k$ only.
\end{remark}

We first prove the following upper bound on $\E[\tau(\mathbf{G}_{\mathrm{simp}},\mathcal{P})]$.

\begin{proposition}\label{prop:balrate}
Fix integer $M$. Let $\mathcal{P}$ be a balanced ordered $M$-partition of $[\bra{k}]$. Then, we have
$$
M\E[\tau(\mathbf{G}_{\mathrm{simp}},\mathcal{P})] \leq k+c_{q,M}+o(1),
$$
as $k\rightarrow\infty$, where $c_{q,M}=O(M)$ is independent of $k$.
\end{proposition}

\begin{IEEEproof}
Note that a round fails to span $\Fq^k$ only if its draws lie in a common hyperplane, so for any layout and any $t$, a union bound over all the $\bra{k}$ hyperplanes gives $\Pr(\tau(\mathbf{G}_{\mathrm{simp}},\mathcal{P})>t)\leq \sum_{H}\rho_H^{t}$. 

For balanced partitions, since $n_i=\bra{k}/M$ for every $i\in [M]$, the AM-GM inequality gives 
$$
\rho_H=\prod_{i=1}^{M}\frac{m_i^{H}}{\bra{k}/M}\leq \parenv{\frac{1}{M}\sum_i\frac{m_i^{H}}{\bra{k}/M}}^M=\parenv{\frac{\bra{k-1}}{\bra{k}}}^M.
$$

Denote $\beta\triangleq [k-1]_q/[k]_q$. We have $\ln(1/\beta)=\ln q+O(q^{-k})$ and $\ln ([k]_q)=(k-1)\ln q+\ln\tfrac{q}{q-1}+O(q^{-k})$. Let $t^*\triangleq \frac{\ln(\bra{k})}{M\ln(1/\beta)}$. Hence
\[ M t^{*}= \frac{\ln ([k]_q)}{\ln(1/\beta)} =(k-1)+\frac{\ln\frac{q}{q-1}}{\ln q}+o(1).\]
Then, we have  $\Pr\parenv{\tau(\mathbf{G}_{\mathrm{simp}},\mathcal{P})>t}\leq \bra{k}\beta^{Mt}$,
which is less than $1$ when $t>t^{*}$.

Since $\E[\tau(\bbG_{\mathrm{simp}},\cP)]=\sum_{t=0}^{\infty}\Pr(\tau(\mathbf{G}_{\mathrm{simp}},\mathcal{P})>t)$, then by bounding the first $\ceilenv{t^{*}}$ tail terms by $1$ and the remaining ones by a geometric series of ratio $\beta^{M}$, and using $\beta\to q^{-1}$, we get 
\[ M \E[\tau(\bbG_{\mathrm{simp}},\cP)] \leq M\ceilenv{t^*} +\frac{M\bra{k}\beta^{M\ceilenv{t^*}}}{1-\beta^{M}} \leq k+c_{q,M}+o(1),\]
where $c_{q,M}\triangleq M-1+\frac{\ln\frac{q}{q-1}}{\ln q}+\frac{M}{1-q^{-M}}$ is independent of $k$, which completes the proof.
\end{IEEEproof} 
 
\begin{remark}
Since $\E[\cC_{\mathrm{simp}}]<k+\delta_q<k+2$ by Remark~\ref{rem:small-overhead}, the bound of Proposition~\ref{prop:balrate} does not exhibit any saving over the single-container model. It demonstrates that the total-read overhead of an \emph{arbitrary} balanced layout remains $O(M)$ uniformly in $k$, i.e., that it does not grow with the dimension.
\end{remark}

To quantify the saving in terms of the expected total number of reads achieved by the layout $(\mathbf{G}_{\mathrm{simp}}, \mathcal{P})$ when $\mathcal{P}$ is balanced, we consider the case where $\mathcal{P}$ is a random balanced partition.

Let $\mathcal{P}_{\mathrm{rb},k} = (P_1, \dots, P_M)$ be a uniformly distributed balanced $M$-partition of the columns of $\mathbf{G}_{\mathrm{simp},k}$. Clearly, $|P_1| = \cdots = |P_M| = \frac{\gv{k}}{M}$. We define the average coverage depth over all possible partitions as
\begin{equation}\label{eq:avg-exp-depth-def}
\overline{E}_{\mathrm{rb}}(k,M) \triangleq \mathbb{E}_{\mathcal{P}_{\mathrm{rb},k}} \left[ \mathbb{E} \left[ \tau\left(\mathbf{G}_{\mathrm{simp},k}, \mathcal{P}_{\mathrm{rb},k}\right) \right] \right].
\end{equation}
We then have the following result.

\begin{theorem}\label{thm:random-balanced}
Let $\tau_{\mathrm{simp},k}$ denote the single-container coverage time of the $k$-dimensional simplex code. Then, as $k\to\infty$ through dimensions satisfying $M\mid\gv{k}$,
$$
M\overline E_{\mathrm{rb}}(k,M) = \E[\tau_{\mathrm{simp},k}] + \gamma_M(k) + o(1) = k+\delta_q+\gamma_M(k)+o(1),
$$
where $\gamma_M(k) \triangleq \E\!\left[[-\tau_{\mathrm{simp},k}]_M\right]$, and $[a]_M$ denotes the least nonnegative residue of an integer $a$ modulo $M$.

Moreover, let $\mathset{\Gamma_j}_{j\geq 1}$ be independent geometric random variables with $\Gamma_j\sim \mathrm{Geom}(1-q^{-j})$ and let $D_\infty\triangleq \sum_{j\geq 1}(\Gamma_j-1)$. For every $a\in\mathset{0,1,\dots,M-1}$ that is attained by infinitely many admissible dimensions (see Remark~\ref{rem:balanced-existence}),
\[ \gamma_M(k) \rightarrow \gamma^{\infty}_{M,a} \triangleq \E\sparenv{[-(a+D_\infty)]_M} 
\qquad\text{as }k\to\infty\text{ with }k\equiv a (\bmod M). \]
Hence $\gamma_M(k)$ is asymptotically periodic with period $M$ on the set of admissible dimensions. If all $M$ residue classes are attained, then, 
one has $\frac{1}{M}\sum_{a=0}^{M-1}\gamma^{\infty}_{M,a}=\frac{M-1}{2}$. 
\end{theorem}

The proof of Theorem~\ref{thm:random-balanced} is technical and combines total-variation bounds, tail truncation, and a limiting geometric-stage decomposition. To maintain the flow of the main text, the full proof is deferred to Appendix~\ref{app:random-balanced-proof}. It is worth to mention that, in general, only the limits $\gamma^{\infty}_{M,a}$ with $a$ attained carry information about the model. For example, for $q=3$ and $M=2$ only $a=0$ is attained, and $\gamma^{\infty}_{2,0}=\Pr(D_\infty\text{ odd}) =\frac{1}{2}\parenv{1-\prod_{j\geq 1}\frac{3^{j}-1}{3^{j}+1}}=0.321$, rather than the unconditional average $\frac{1}{2}$.

\section{Conclusion and Open Problems}
\label{sec:conclusion}

This paper establishes a general framework for analyzing full-message recovery in distributed DNA storage via the notion of distributed coverage depth. First, we provide exact formulas, general lower bounds, and optimal layout constructions for linear codes, proving that MDS codes remain optimal in minimizing the expected number of sequencing rounds. 

Second, a systematic evaluation of MDS codes reveals a key trade-off: while constant-size balanced containers achieve a genuine reduction in total reads compared to the single-container case, polynomial-size containers merely accelerate recovery via parallelism without reducing the total read volume. Furthermore, our low-rate lower bounds and higher-rate constructions show that these distributional gains are governed jointly by code rate and container scaling. 

Third, for simplex codes, we resolve a conjecture by Bertuzzo, Ravagnani, and Yaakobi~\cite{BRY26} in the single-container setting, construct optimal staircase layouts achieving minimal total reads when $M=k$, and bound the total-read overhead for both deterministic and random balanced layouts.

Together, these results demonstrate that distribution is neither inherently beneficial nor universally detrimental. Ultimately, whether distribution reduces coverage overhead or merely provides parallel speedup depends fundamentally on the underlying rank geometry of the code and the combinatorial structure of the layout.

Several natural directions and open problems emerge from our results. Here we highlight a few of them for future research.

\begin{enumerate}

\item \emph{Optimal layouts under given constraints.}
For a fixed generator matrix $\mathbf{G}$, a fixed number of containers $M$, and prescribed container sizes $n_1,\ldots,n_M$, determine
$$
\min_{\mathcal{P}}
\E[\tau(\mathbf{G},\mathcal{P})].
$$
Even when the partition is required to be balanced, a structural characterization of the optimal layout is not known for a general code. It would be particularly useful to characterize which codes admit balanced one-round-perfect layouts and, when this is impossible, to determine the smallest achievable number of total reads. A related question is to develop nearly balanced constructions when $M\nmid n$.

\item \emph{Optimal codes over small fields.}
Given a field $\mathbb{F}_q$, a fixed integer $M$, and a fixed ordered $M$-partition $\mathcal{P}$ of $[n]$, determine
$$
\min_{\mathbf{G}} \mathbb{E}[\tau(\mathbf{G},\mathcal{P})],
$$
and identify the optimal rank-$k$ generator matrix $\mathbf{G}$. This problem represents the counterpart to the previous one, focusing on finding optimal codes that minimize the distributed coverage depth. While this problem is completely settled when MDS codes exist, it remains open for parameters where no MDS code is available (e.g., over small fields where $q < n - 3$). In the single-container setting, the analogous question was studied in \cite{BRY26} using duality arguments and extended weight enumerators; adapting such combinatorial tools may provide a fruitful direction for the distributed storage model.

\item \emph{Simplex code layouts with intermediate numbers of containers.}
The integrality arguments rule out savings for many values of $q,k,$ and $M$, and the codimension-two analysis settles two small exceptional cases. It remains to determine the optimal distributed coverage depth in the remaining regimes, especially when $M\mid k$, and in the binary cases where $M\mid(k+1)$. 

\item \emph{Nonuniform, asynchronous, adaptive, and noisy sampling.}
The present model assumes that all strands within a container are equally likely to be sampled, that one read is obtained from every container in every round, and that every sampled strand is observed without error. Practical systems may exhibit unequal strand abundances, different sequencing rates across containers, adaptive allocation of reads, erasures, or sequencing errors. Extending the present results to these models, and determining which of the layout-optimality principles remain valid, are natural directions for future work.

\end{enumerate}

\appendices

\section{One-Round-Perfect Constructions and Examples}\label{app:one-round-perfect additional}



\begin{theorem}\label{thm:circuit}
Let $\mathcal{C}$ be an $[n,k]$ linear code with redundancy $r = n - k \ge 1$ and parity-check matrix $\mathbf{H}$. Suppose some $r+1$ columns of $\mathbf{H}$, indexed by $B \subseteq [n]$, are in general position\footnote{A set of $\ell$ vectors in $\F_q^{m}$ with $\ell\geq m$ is said to be \emph{in general position} if every $m$ of them are linearly independent.}.
Then, the generator matrix $\mathbf{G}$ and the $k$-partition 
\begin{equation}\label{eq:partition-gen-pos}
\mathcal{P} = (B,\{j_1\},\{j_2\},\ldots,\{j_{n-|B|}\}),
\end{equation}
where $[n]\setminus B=\{j_1,j_2,\ldots,j_{n-|B|}\}$ together yield a one-round-perfect layout.
\end{theorem}

\begin{IEEEproof}
Recall that a parity-check matrix of an $[n,k]$ code $\mathcal{C}$ is any full-rank matrix $\mathbf{H}=(\mathbf{h}_1,\ldots,\mathbf{h}_n)\in \mathbb{F}_q^{r\times n}$ such that $\mathcal{C}=\ker (\mathbf{H})$.

First, by definition, the partition has $(n-r-1)+1 = n-r = k$ non-empty parts and is clearly a $k$-partition.

Let $T\subseteq[n]$ be an index set such that $([n]\setminus B)\subseteq T$ and $|T\cap B|=1$. Thus, its complement $T^C = [n]\setminus T$ has the form $B\setminus\{b\}$ and is of size $r$. Suppose that the columns of the generator matrix indexed by $T$ fail to span $\mathbb{F}_q^{k}$. Then, there exists a non-zero vector $\mathbf{h}\in \mathbb{F}_q^{k}$ such that the codeword $\mathbf{c}=\mathbf{h}\mathbf{G}$ is non-zero and has its support contained in $[n]\setminus T = B\setminus\{b\}$. However, the parity-check condition $\mathbf{H}\mathbf{c}^{\top}=\mathbf{0}$ reads $\sum_{j\in\mathrm{supp}(\mathbf{c})} c_j\mathbf{h}_j = \mathbf{0}$. This yields a non-trivial linear dependence among at most $r$ columns of $\mathbf{H}$ indexed by $B$, contradicting the condition that the columns indexed by $B$ are in general position. Hence, the columns of the generator matrix $\mathbf{G}$ indexed by every such $k$-set $T$ span $\mathbb{F}_q^k$, which proves that $(\mathbf{G},\mathcal{P})$ is one-round-perfect.
\end{IEEEproof}

According to the proof of Theorem~\ref{thm:circuit}, since $\mathbf{H}$ has rank $r$, any $r+1$ of its columns are linearly dependent. Thus, by the definition of general position, every dependency inside $B$ must include all $r+1$ columns. This implies that there exists a codeword of $\mathcal{C}$ with support set $B$, while no non-zero codeword of $\mathcal{C}$ is supported on a proper subset of $B$. From this observation, we draw two consequences and one caution. First, if $\mathcal{C}$ is an MDS code, then so is its dual $\mathcal{C}^{\perp}$; hence, every $r$ columns of $\mathbf{H}$ are linearly independent, and by Theorem~\ref{thm:circuit}, \emph{every} $(r+1)$-subset $B$ induces a one-round-perfect $k$-partition (see Theorem~\ref{thm:mdssplit} in Section~\ref{sec: MDS distributed coverage} below). The second consequence is the following result, which shows that Hamming codes also admit one-round-perfect $k$-partitions.

\begin{corollary}\label{cor:hammingfloor}
For every prime power $q$ and every integer $r\ge 2$, let $\mathcal{C}_{\mathrm{Hamming}}$ be the $[\bra{r},\bra{r}-r]$ Hamming code over $\mathbb{F}_q$ whose parity-check columns represent the $\bra{r}$ distinct directions of $\mathbb{F}_q^{r}$. Let $B$ be the set of those $r+1$ indices whose corresponding parity-check columns are the $r$ standard basis vectors $\mathbf{e}_1,\ldots,\mathbf{e}_r$ and their sum $\mathbf{e}_1+\cdots+\mathbf{e}_r$. Then the layout $(\mathbf{G}_{\mathrm{Hamming}},\mathcal{P})$ with $\mathcal{P} = (B,\{j_1\},\{j_2\},\ldots,\{j_{n-|B|}\})$, where $\{j_1,j_2,\ldots,j_{n-|B|}\}=[\bra{r}]\setminus B$ is one-round-perfect layout.
\end{corollary}

\begin{IEEEproof}
According to Theorem~\ref{thm:circuit}, it suffices to verify that the $r+1$ columns indexed by $B$ are in general position. This follows directly from their specific forms.
\end{IEEEproof}

The caution is that the condition in Theorem~\ref{thm:circuit} is sufficient but not necessary. For instance, the binary first-order Reed--Muller code $\mathrm{RM}(1,3)$ is self-dual, and its parity-check matrix has \emph{no} $5$-subset of parity-check columns in general position; nevertheless, it admits a one-round-perfect $4$-partition.

To support the above discussion, we next focus on first-order Reed--Muller codes. Unlike MDS codes and Hamming codes, which can admit high code rates, the first-order Reed--Muller code is a low-rate code family, similar to simplex codes.

\begin{theorem}\label{thm:rm}
For $\mathcal{C}=\mathrm{RM}_q(1,m)$, the following two statements hold.
\begin{enumerate}
\item[(i)] The coverage depth of $\mathcal{C}$ in the single-container model is
$$
\mathbb{E}[\mathcal{C}] = k + \sum_{j=1}^{m}\frac{1}{q^{j}-1},
$$
so the overhead is $\mathrm{Ov}(\mathcal{C}) = \sum_{j=1}^{m}(q^{j}-1)^{-1}$.
\item[(ii)] The \emph{affine staircase} $k$-partition of $[n]=[q^{m}]$, consisting of $P_0=\{j\in [q^m] \colon \mathbf{p}_j=\mathbf{0}\}$ and, for $1\le i\le m$, 
$$
P_i = \{j\in [q^m]:~ \mathbf{p}_j=(p_{j,1},\ldots,p_{j,m}) \text{ such that } p_{j,1}=\dots=p_{j,i-1}=0,\ p_{j,i}\neq 0\}
$$
of size $(q-1)q^{m-i}$, is one-round-perfect.
\end{enumerate}
\end{theorem}

\begin{proof}
The proof of both parts relies on a key observation. A linear combination $\sum_i \lambda_i (1,\mathbf{p}_i)$ of columns has $1$ as its first coordinate if and only if $\sum_i \lambda_i = 1$; more generally, it is a non-zero scalar multiple of some column $(1,\mathbf{a})$ if and only if $\sum_i \lambda_i \neq 0$. Thus, the vector $\mathbf{a}$ can be viewed as an \emph{affine} combination of $\{\mathbf{p}_i\}_{i\in [q^m]}$. 

For any non-empty subset $T$ of columns, let $X_T \triangleq \{\mathbf{p}_i \colon (1,\mathbf{p}_i)\in T\}$. By fixing any $\mathbf{p}_0 \in X_T$, we have $\spn(T) = \spn\{(1,\mathbf{p}_0)\} \oplus \spn\{(0,\mathbf{p}_i - \mathbf{p}_0) \colon \mathbf{p}_i \in X_T\}$. This yields
\begin{equation}\label{eq:affinehull}
\begin{cases}
    \spn(T) \cap \{(1,\mathbf{p}_i)\}_{i\in [q^m]} = \{(1,\mathbf{a}) \colon \mathbf{a} \in \mathrm{aff}(X_T)\}; \\[1ex]
    \dim(\spn(T)) = 1 + \dim(\mathrm{aff}(X_T)),
\end{cases}
\end{equation}
where $\mathrm{aff}(X_T)$ denotes the affine hull of $X_T$ in $\mathbb{F}_q^{m}$. In particular, $\spn(T)=\mathbb{F}_q^{k}$ if and only if $\mathrm{aff}(X_T)=\mathbb{F}_q^{m}$.
 
For (i), note that this is the single-container model. Let $V^{(i)}$ be the span of the sampled columns at the first time its dimension becomes $i$ as in Corollary~\ref{cor:single-container-stage-decomposition}. By \eqref{eq:affinehull}, $V^{(i)}$ corresponds to an affine subspace of dimension $i-1$ for $i\ge1$, which contains exactly $q^{i-1}$ vectors in $\Fq^{m}$. Hence, $V^{(i)}$ holds exactly $q^{i-1}$ of the $q^{m}$ columns, which yields $\pi(V^{(i)})=q^{i-1-m}$ is deterministic. Hence, by \eqref{eq:single-container-stage-formula} in Corollary~\ref{cor:single-container-stage-decomposition}, we have
$$
\E[\mathcal{C}]=1+\sum_{i=1}^{m}\frac1{1-q^{i-1-m}}=1+\sum_{j=1}^{m}\frac1{1-q^{-j}}
=1+\sum_{j=1}^{m}\Bigl(1+\frac1{q^{j}-1}\Bigr)=k+\sum_{j=1}^{m}\frac1{q^{j}-1}.
$$
 
For (ii), we first note that $\{P_i\}_{i=0}^{m}$ forms a $k$-partition of $[q^{m}]$. The pairwise disjointness of the sets $P_i$ and their corresponding sizes follow directly from their definitions. Specifically, $\mathbf{p}_j=\mathbf{0}$ lies in $P_0$, and any $\mathbf{p}_j\neq \mathbf{0}$ belongs to the part $P_i$ corresponding to its first non-zero coordinate $i$. Hence, $\bigcup_{i=0}^m P_i = [q^m]$.

Next, we show that the partition $\{P_i\}_{i=0}^{m}$ is one-round-perfect. Let $\{(1,\mathbf{x}^{(0)}),(1,\mathbf{x}^{(1)}),\dots,(1,\mathbf{x}^{(m)})\}$ be a set of drawn vectors within one time round, with $\mathbf{x}^{(0)}=\mathbf{0} \in P_0$ and $\mathbf{x}^{(i)}\in P_i$ for $1\le i\le m$. By \eqref{eq:affinehull}, this set spans $\mathbb{F}_q^{k}$ if and only if vectors $\mathbf{x}^{(0)},\ldots,\mathbf{x}^{(m)}$ affinely span $\mathbb{F}_q^{m}$, which is equivalent to $\{\mathbf{x}^{(i)}-\mathbf{x}^{(0)} = \mathbf{x}^{(i)}\}_{i\in [m]}$ being linearly independent. This follows directly because the $m \times m$ matrix with rows $\mathbf{x}^{(1)},\dots,\mathbf{x}^{(m)}$ is upper triangular with non-zero diagonal entries, since each $\mathbf{x}^{(i)}$ has its first non-zero coordinate precisely at position $i$ and is therefore invertible. Hence, the set $\{(1,\mathbf{x}^{(0)}),(1,\mathbf{x}^{(1)}),\dots,(1,\mathbf{x}^{(m)})\}$ spans $\mathbb{F}_q^{k}$, which confirms that the partition $\{P_i\}_{i=0}^{m}$ is one-round-perfect. This completes the proof.
\end{proof} 

\section{Technical Results for the MDS Analysis}\label{app:mds-technical}

\subsection{Concentration Results for Negatively Associated Random Variables}

We first recall the following definition and properties of negatively associated random variables.

\begin{definition}[\!\!{\cite{Dubhashi_Panconesi_2009}}]\label{def: NA}
Random variables $X_1, X_2, \ldots, X_n$ are said to be negatively associated if for every disjoint index sets $I, J \subseteq [n]$ and coordinate-wise non-decreasing functions $f$ and $g$, it holds that $\mathrm{Cov}\bigl(f(X_I), g(X_J)\bigr) \leq 0$, where $X_I = (X_i : i \in I)$ and $X_J = (X_j : j \in J)$.
\end{definition}

\begin{theorem}[Chernoff--Hoeffding bounds for NA variables {\cite{Hoeffding63,DR98}}]\label{thm: CH bound for NA variables}
Let $X_1, X_2, \ldots, X_m$ be negatively associated random variables with $X_j \in [a_j,b_j]$ always. Then $Y = \sum_{j=1}^{m} X_j$ satisfies the following bound
\begin{equation}
\Pr(|Y-\mathbb{E}[Y]| \ge t) \le 2\exp\left( -\frac{2t^2}{\sum_{j=1}^{m} (b_j-a_j)^2} \right).
\end{equation}
\end{theorem}

\subsection{Balanced Two-Container MDS Examples with Reduced Total Reads}

\begin{example}\label{prop:mdstwo}
In this example, we provide three sets of parameters for which there exist MDS codes with balanced $2$-partitions under the distributed model that require strictly fewer total reads than in the single-container model. All three parameter sets can be realized by Reed--Solomon codes.

\begin{itemize}
    \item Case 1: $n=6$ , $k=4$. Recall that $D_i(t)$ denote the number of distinct columns observed from container $i$ up to round $t$. Hence, $\tau_k>t$ if and only if $(D_1(t),D_2(t))\in\{(1,1),(1,2),(2,1)\}$. Thus,
    \begin{align*}
        \Pr(\tau_k>t) & = \Pr(D_1(t)=1)\bigl(\Pr(D_2(t)=1)+\Pr(D_2(t)=2)\bigr) \\
        &~~~ +\Pr(D_1(t)=2)\Pr(D_2(t)=1) \\
        & = 9^{1-t}(2^{t+1}-3)=18(\tfrac29)^{t}-27(\tfrac19)^{t},
    \end{align*}
    where the last equality follows by \eqref{eq:distinctlaw} with $n_i=3$. Summing the geometric series yields
    $$
    \mathbb{E}[\tau_k] = 1 + \sum_{t\ge 1} \Pr(\tau_k > t) = 1 + 18 \cdot \frac{2/9}{7/9} - 27 \cdot \frac{1/9}{8/9} = \frac{155}{56}.
    $$
    On the other hand, $\mathbb{E}[\mathcal{C}] = 6(H_6 - H_2) = \tfrac{57}{10} = 5.7$, which is strictly larger than $2\mathbb{E}[\tau_k] \approx 5.5357$.
    \item Case 2: $n=8$, $k=4$. Similarly, we have $\tau_k > t$ if and only if $(D_1(t),D_2(t))\in\{(1,1),(1,2),(2,1)\}$. Using \eqref{eq:distinctlaw} with $n_i=4$, the same computation yields $\Pr(\tau_k>t)=48(\tfrac18)^{t}-80(\tfrac1{16})^{t}$ and $\mathbb{E}[\tau_k]=1+\tfrac{48}{7}-\tfrac{16}{3}=\tfrac{53}{21}$. Thus, the expected total number of reads is $2\mathbb{E}[\tau_k] = \tfrac{106}{21}=\tfrac{530}{105}$, which improves upon $\mathbb{E}[\mathcal{C}]=8(H_8-H_4)=\tfrac{533}{105}$.
    \item Case 3: $n=6$, $k=5$. Since $k=5$, we have $\tau_k> t$ if and only if $(D_1(t),D_2(t))\in\{(1,1),(1,2),(2,1),(1,3),$ $(3,1),(2,2)\}$. Using \eqref{eq:distinctlaw} with $n_i=3$, by similar computations to those in the previous cases, we have
    $$
    \Pr(\tau_k>t)=6(\tfrac13)^{t}+9(\tfrac49)^{t}-36(\tfrac29)^{t}+27(\tfrac19)^{t},
    $$
    which leads to $\E[\tau_k]=1+6\cdot \frac{1}{2}+9\cdot \frac{4}{5}-36\cdot \frac{2}{7}+27\cdot \tfrac{1}{8}=\tfrac{1201}{280}$. Hence, the total number of reads is $2\E[\tau_k]=\frac{1201}{140}$ which improves upon $\E[\mathcal{C}]=6(H_6-H_1)=\tfrac{87}{10}=\tfrac{1218}{140}$.
\end{itemize}
\end{example}

\section{Exceptional Two-Container Cases for Simplex Codes}\label{app:simplex-exceptional}

By Theorem~\ref{thm:intclass}, the only parameter regimes in which the distributed model with $M$ ($1<M<k$) containers could potentially yield a smaller expected total number of reads are those where $M \mid k$ (for any $q$), and $M \mid (k+1)$ with $q=2$ and $k \ge 4$. In the following, we consider the two simplest parameter settings in these regimes for $M=2$, namely, $(q,k)=(2,5)$ and $(3,4)$. We show that for $M=2$, the distributed model cannot improve upon the single-container model in terms of the expected total number of reads.
 
\begin{theorem}\label{thm:simplexm2}
Let $\mathcal{C}$ be a $q$-ary simplex code of dimension $k$ with standard generator matrix $\mathbf{G}_{\mathrm{simp}}$, and let $\mathcal{P}$ be any $2$-partition of $[(q^k-1)/(q-1)]$. If $(q,k)=(2,5)$ or $(3,4)$, then $2\E[\tau(\mathbf{G}_{\mathrm{simp}},\mathcal{P})] > \E[\mathcal{C}]$.
\end{theorem}

For the proof of Theorem~\ref{thm:simplexm2}, we need some preliminary results. First, we present the following corollary of Theorem \ref{thm:distributed-exact-representations} for general $M$.

\begin{corollary}\label{cor:codim-two-bonferroni}
Under the assumptions of
Theorem~\ref{thm:distributed-exact-representations}, let $u_t \triangleq k-\dim(V_t)$, and define $X_t \triangleq \sum_{H\lneq \Fq^{k}}\rho_H^t$, $Y_t \triangleq \sum_{\mathrm{codim}(U)=2}\rho_U^t$. Then, for every integer $t\geq0$,
\begin{equation}\label{eq:codim-two-exact-remainder}
    \Pr\bigl(\tau(\mathbf G,\mathcal P)>t\bigr) = X_t-qY_t + q^3
    \mathbb E\left[
        \mathbf 1_{\{u_t\geq3\}}
        \qbin{u_t-1}{2}
    \right].
\end{equation}
In particular, $\Pr\bigl(\tau(\mathbf G,\mathcal P)>t\bigr) \geq X_t-qY_t$.
\end{corollary}

\begin{IEEEproof}
For $1\leq d\leq k$, put $Z_{d,t}\triangleq \sum_{\mathrm{codim}(U)=d}\rho_U^t$. By~\eqref{eq:mobius-tail-formula},
\begin{equation}\label{eq:mobius-grouped-by-codim}
    \Pr\bigl(\tau(\mathbf{G},\mathcal{P})>t\bigr)
    =
    \sum_{d=1}^k
    (-1)^{d-1}q^{\binom d2}Z_{d,t}.
\end{equation}

Note that for each subspace $H\lneq \Fq^k$, $\Pr(V_t\subseteq H)=\rho_H^{t}$, and the number of codimension-$d$ subspaces containing $V_t$ is $\qbin{u_t}{d}$. Hence, $Z_{d,t} =\sum_{\mathrm{codim}(U)=d}\Pr(V_t\subseteq U)= \mathbb E\left[\qbin{u_t}{d}\right]$. In particular, $Z_{1,t}=\mathbb E[\bra{u_t}]=X_t$ and $Z_{2,t}=\mathbb E\left[\qbin{u_t}{2}\right]=Y_t$. Moreover, by the linearity of expectation, we have
\begin{align}
\sum_{d=3}^{k}
(-1)^{d-1}q^{\binom d2}Z_{d,t} & =
\mathbb{E}\left[
    \sum_{d=3}^{k}
    (-1)^{d-1}q^{\binom d2}
    \qbin{u_t}{d}
\right] \nonumber \\
& = \mathbb{E}\left[
    \sum_{d=3}^{u_t}
    (-1)^{d-1}q^{\binom d2}
    \qbin{u_t}{d}
\right], \label{eq: change of upper limits of summation}
\end{align}
where the second equality follows since for every realization of $u_t$, one has $\qbin{u_t}{d}=0$ whenever $d>u_t$.

For every integer $u\geq1$, the \emph{$q$-binomial theorem} \cite[Equation 1.87]{Stanley2011enumerative} gives $1 = \sum_{d=1}^{u} (-1)^{d-1}q^{\binom d2}\qbin{u}{d}$. Separating the terms with $d=1$ and $d=2$, we obtain
\begin{align}\label{eq:mobius-two-term-remainder}
    \sum_{d=3}^{u} (-1)^{d-1}q^{\binom d2}\qbin{u}{d} = 1-\bra{u}+q\qbin{u}{2} = q^3\qbin{u-1}{2}
\end{align}
for every $u\geq1$, where the second equality follows by $\bra{u}=1+q\bra{u-1}$ and the $q$-Pascal identity $\qbin{u}{2} = q^2\qbin{u-1}{2}+\bra{u-1}$. Hence, by substituting \eqref{eq:mobius-two-term-remainder} into \eqref{eq:mobius-grouped-by-codim} and taking expectations, \eqref{eq: change of upper limits of summation} yields
$$
\Pr\bigl(\tau(\mathbf G,\mathcal P)>t\bigr) = X_t-qY_t + q^3
    \mathbb E\left[
        \mathbf 1_{\{u_t\geq3\}}
        \qbin{u_t-1}{2}
    \right].
$$
Since the last term is nonnegative, $\Pr\bigl(\tau(\mathbf G,\mathcal P)>t\bigr) \geq X_t-qY_t$ follows.
\end{IEEEproof}

Next, we establish some lower bounds on the partial sums of $X_t$ and $Y_t$ defined in Corollary~\ref{cor:codim-two-bonferroni}.

\begin{lemma}\label{lem:simplex-two-container-moments}
Let $\mathbf{G}_{\mathrm{simp}}$ be the standard generator matrix of the $q$-ary simplex code of dimension $k$, and let $\mathcal{P}$ be a $2$-partition of $[(q^k-1)/(q-1)]$. Let $X_t$ and $Y_t$ be as defined in Corollary~\ref{cor:codim-two-bonferroni}. Let $\mathcal{A}$ be any nonempty family of hyperplanes, and define
$\overline{\rho}_{\mathcal{A}} \triangleq \frac{1}{|\mathcal{A}|} \sum_{H\in\mathcal{A}}\rho_H$. Then, for every integer $t_0 \ge 1$,
\begin{equation}\label{eq:subfamily-jensen-tail}
    \sum_{t\ge t_0} \sum_{H\in\mathcal{A}} \rho_H^t
    \ge
    |\mathcal{A}|
    \frac{\overline{\rho}_{\mathcal{A}}^{\,t_0}}{1 - \overline{\rho}_{\mathcal{A}}}.
\end{equation}
In particular,
\begin{equation}\label{eq:full-jensen-tail}
    \sum_{t\ge t_0} X_t
    \ge
    [k]_q
    \frac{\left([k-2]_q / [k]_q\right)^{t_0}}{1 - [k-2]_q / [k]_q}.
\end{equation}
Moreover, defining $\mu_* \triangleq \max_{\mathrm{codim}(U) = 2} \rho_U$, we have $\mu_* < 1$ and
\begin{equation}\label{eq:Y-tail-general}
    \sum_{t\ge t_0} Y_t
    \le
    \frac{\mu_*^{\,t_0-1}}{1 - \mu_*}
    \begin{bmatrix} k-2 \\ 2 \end{bmatrix}_q.
\end{equation}
Finally,
\begin{equation}\label{eq:mu-star-combinatorial}
    \mu_* \le \frac{\lfloor [k-2]_q / 2 \rfloor \lceil [k-2]_q / 2 \rceil}{n_1 n_2}.
\end{equation}
\end{lemma}

\begin{IEEEproof}
For every fixed integer $t\geq1$, the function $f(x)= x^t$ is convex on $[0,1]$. Hence, by Jensen's inequality,
$$
\sum_{H\in\mathcal{A}}\rho_H^t \geq
    |\mathcal{A}|
    \left( \frac{1}{|\mathcal{A}|} \sum_{H\in\mathcal{A}}\rho_H
    \right)^t 
    = |\mathcal{A}|\,\overline\rho_{\mathcal{A}}^{\,t}.
$$
Since every hyperplane is a proper subspace, one has $\rho_H<1$ for every hyperplane $H$. In particular,
$\overline\rho_{\mathcal{A}}<1$. Summing the preceding inequality over
$t\geq t_0$ proves \eqref{eq:subfamily-jensen-tail}.

According to \eqref{eq:rho}, we have
\begin{align*}
    \sum_{H\lneq \Fq^k}\rho_H & = \sum_{H\lneq \Fq^k}\sum_{\mathbf{g}_1\in \mathbf{G}_1}\sum_{\mathbf{g}_2\in \mathbf{G}_2}\frac{\mathbf{1}\{\mathbf{g}_1,\mathbf{g}_2\in H\}}{n_1n_2} \\
    & = \sum_{\mathbf{g}_1\in \mathbf{G}_1}\sum_{\mathbf{g}_2\in \mathbf{G}_2}\frac{|\{H\lneq \Fq^k:\spn\{\mathbf{g}_1,\mathbf{g}_2\}\subseteq H\}|}{n_1n_2} \\
    & = \bra{k-2},
\end{align*}
where the last equality follows since columns of $\mathbf{G}_{\mathrm{simp}}$ are projective points in $\mathbb{P}$ and any two of them are linearly independent. Note that the number of hyperplanes in $\Fq^k$ is $\bra{k}$. Hence, \eqref{eq:full-jensen-tail} follows by applying \eqref{eq:subfamily-jensen-tail} to the family of all hyperplanes.

We next consider $Y_t$. Since  $\rho_U<1$ for every codimension-two subspace $U$, we have $\mu_*<1$. Therefore,
\begin{align*}
    \sum_{t\geq t_0}Y_t = \sum_{\mathrm{codim}(U)=2}
    \frac{\rho_U^{t_0}}{1-\rho_U}  \sum_{\mathrm{codim}(U)=2} \rho_U \frac{\rho_U^{t_0-1}}{1-\rho_U} \leq
    \frac{\mu_*^{\,t_0-1}}{1-\mu_*}\sum_{\mathrm{codim}(U)=2}\rho_U,
\end{align*}
where the last inequality follows since $\frac{x^{t_0-1}}{1-x}$ is increasing in $x$ for $0<x<1$ and $t_0\geq 1$. Again, according to \eqref{eq:rho}, we have
$$
\sum_{\mathrm{codim}(U)=2}\rho_U =
    \sum_{\mathbf{g}_1\in \mathbf{G}_1}\sum_{\mathbf{g}_2\in \mathbf{G}_2}\frac{|\{\mathrm{codim}(U)=2:\spn\{\mathbf{g}_1,\mathbf{g}_2\}\subseteq U\}|}{n_1n_2} = \qbin{k-2}{2},
$$
and hence \eqref{eq:Y-tail-general} follows.

Finally, every codimension-two subspace contains exactly $\bra{k-2}$ projective points. Denote $a_U\triangleq |\mathbf{G}_1\cap U|$, then, we have $|\mathbf{G}_2\cap U|=\bra{k-2}-a_U$. Thus, $\rho_U=\frac{a_U(\bra{k-2}-a_U)}{n_1n_2}$. Set $b\triangleq\min\mathset{n_1,\floorenv{[k-2]_q/2}}$. Since $a_U=|\bbG_1\cap U|\leq \min\mathset{n_1,[k-2]_q}$ and $a\mapsto a([k-2]_q-a)$ is unimodal, we obtain 
\begin{equation}\label{eq:mu-star}
\rho_U\leq \frac{b\parenv{[k-2]_q-b}}{n_1n_2}\leq \frac{\lfloor\bra{k-2}/2\rfloor\lceil\bra{k-2}/2\rceil}{n_1n_2}.
\end{equation}
Taking the maximum over $U$ proves
\eqref{eq:mu-star-combinatorial}.
\end{IEEEproof}

Next, using Lemma~\ref{lem:simplex-two-container-moments}, we proceed to prove Theorem~\ref{thm:simplexm2}.

\begin{IEEEproof}[Proof of Theorem \ref{thm:simplexm2}]
We identify the columns of $\mathbf{G}_{\mathrm{simp}}$ with the set of projective points $\mathbb{P}$ in $\Fq^{k}$. Put $n=\bra{k}$, and $n_i=|P_i|$, $i=1,2$. By the symmetry between the two containers, assume throughout the proof that $n_1\leq n_2$.

For a vector subspace $U\leq\mathbb{F}_q^k$, recall that $\rho_U =\frac{|C_1\cap U|}{n_1}\frac{|C_2\cap U|}{n_2}$ and $\rho_U$ is the probability that the two columns drawn in one round both lie in $U$. Consequently, the probability that all columns drawn during the first $t$ rounds lie in $U$ is $\rho_U^t$.

We first record a quotient reduction that will be used when the container with smaller size of the two is very small. Let $T_{q,k}$ denote the recovery time of $k$-dim $q$-ary simplex code in the single-container model, i.e., $T_{q,k}=\tau(\mathbf{G}_{\mathrm{simp}},\mathcal{P}_0)$ with $\mathcal{P}_0=\{[(q^k-1)/(q-1)]\}$. 

Let $W\triangleq \mathrm{span}\{\mathbf g_j:j\in \mathbf{G}_1\}$ and suppose that $w\triangleq \dim W\leq k-2$. Consider the quotient space $Q\triangleq \mathbb F_q^k/W$. Clearly, every column drawn from $\mathbf{G}_1$ projects to zero in $Q$. Moreover, every projective point of $Q$ has $\bra{w+1}-\bra{w}=q^w$ projective preimages outside $W$, all of which belong to $\mathbf{G}_2$. Hence, conditioned on lying outside $W$, the projection of a column drawn from $\mathbf{G}_2$ in $Q$ is uniform on the projective points of $Q$.

Let $R_j$ be the round in which the $j$-th drawn column from $\mathbf{G}_2$ lies outside $W$, and let $T_Q$ be the number of such draws required to span $Q$. Then, according to the discussion above, $T_Q$ has the same distribution as $T_{q,k-w}$. Moreover, since all columns drawn from $\mathbf{G}_1$ project to zero in $Q$, we have
$\tau(\mathbf G_{\mathrm{simp}},\mathcal P) \geq R_{T_Q} \geq T_Q$ a.s. Therefore,
\begin{equation}\label{eq:quotient-bound-in-proof}
    \mathbb E[\tau(\mathbf G_{\mathrm{simp}},\mathcal P)]
    \geq
    \mathbb E[T_{q,k-w}].
\end{equation}
If $n_1=1$, then $W$ is the line spanned by the unique projective point of $\bbG_1$, so $w=1$ and every column drawn from $\bbG_2$ lies outside $W$. Since each projective point of $Q=F_q^{k}/W$ has exactly $[2]_q-[1]_q=q$ preimages outside $W$, all of them in $\bbG_2$, the projection of a draw from $\bbG_2$ is uniform on the projective points of $Q$. Consequently $\tau(\bbG_{\mathrm{simp}},\cP)$ has the same distribution as $T_{q,k-1}$.

We now consider the two parameter sets separately. 

\medskip
\noindent
\emph{Case 1: $(q,k)=(2,5)$.}
Here $n=31$, and $\mathbb E[\mathcal C]=\mathbb E[T_{2,5}]=\frac{3567}{560}$ by Theorem \ref{thm: simplexE}.

Since at most four columns have been drawn after two rounds, for $0\leq t\leq 2$, we have$\Pr\bigl(\tau(\mathbf G_{\mathrm{simp}},\mathcal P)>t\bigr)=1$. Therefore, denote
$$
    \Sigma_{2,5}\triangleq \sum_{t\geq3}
    \Pr\bigl( \tau(\mathbf{G}_{\mathrm{simp}},\mathcal P)>t\bigr),
$$
we have $\mathbb E[\tau(\mathbf G_{\mathrm{simp}},\mathcal P)] = 3+\Sigma_{2,5}$. It suffices to prove $\Sigma_{2,5}> \frac{\mathbb E[\mathcal C]-6}{2} = \frac{207}{1120}$.

Suppose first that $n_1\geq4$. In this case, $\bra{k-2}=\bra{3}=7$, $\qbin{k-2}{2}=\qbin{3}{2}=7$. By \eqref{eq:full-jensen-tail}, with $t_0=3$,
$$
\sum_{t\geq3}X_t \geq 31 \frac{(7/31)^3}{1-7/31} = \frac{343}{744}.
$$
By \eqref{eq:mu-star-combinatorial}, $\mu_* \leq \frac{3\cdot4}{n_1(31-n_1)} \leq \frac{12}{4\cdot27} = \frac{1}{9}$. Hence \eqref{eq:Y-tail-general} yields
$$
2\sum_{t\geq3}Y_t \leq 2 \frac{(1/9)^2}{1-1/9} \cdot7 = \frac{7}{36}.
$$
Then, by Corollary~\ref{cor:codim-two-bonferroni}, we obtain
$$
\Sigma_{2,5} \geq \frac{343}{744}-\frac{7}{36} = \frac{595}{2232} > \frac{207}{1120}.
$$

Next suppose that $n_1=3$. The three projective points of $\mathbf{G}_1$ are contained in some hyperplane $H_0$. Since $H_0$ contains $\bra{4}=15$ projective points, we have $\rho_{H_0} = \frac{3}{3}\frac{15-3}{28} =\frac{3}{7}$. Recall that $\sum_{H\lneq \Fq^k}\rho_H=\bra{k-2}=7$. Hence, the average mass of the remaining thirty hyperplanes is $u=\frac{7-3/7}{30}{}=\frac{46}{7\cdot30}=\frac{23}{105}$. Applying \eqref{eq:subfamily-jensen-tail} with $\mathcal{A}$ being the set of the remaining thirty
hyperplanes and $\overline\rho_{\mathcal{A}}=u$ yields
$$
\sum_{t\geq3}X_t \geq
    \frac{(3/7)^3}{1-3/7} + 30\frac{u^3}{1-u}
    =
    \frac{27}{196} + \frac{12167}{30135}.
$$
Moreover, by \eqref{eq:mu-star-combinatorial}, $\mu_*\leq\frac{3\cdot4}{3\cdot28}=\frac{1}{7}$. Hence,
$$
2\sum_{t\geq3}Y_t \leq 2 \frac{(1/7)^2}{1-1/7} \cdot 7 = \frac{1}{3}
$$
by \eqref{eq:Y-tail-general}. Then, it follows that
$$
\Sigma_{2,5} \geq \frac{27}{196} + \frac{12167}{30135} - \frac{1}{3} = \frac{25093}{120540} > \frac{207}{1120}.
$$


If $n_1=2$, then the two distinct projective points of $\mathbf{G}_1$ span a two-dimensional space. Thus \eqref{eq:quotient-bound-in-proof} yields
$$
\mathbb E[\tau(\mathbf G_{\mathrm{simp}},\mathcal P)] \geq \mathbb E[T_{2,3}] = \frac{47}{12} > \frac{3567}{1120} = \frac{\mathbb E[\mathcal C]}{2}.
$$
If $n_1=1$, since $\tau(\mathbf G_{\mathrm{simp}},\mathcal P)$ has the same distribution as $T_{q,k-1}$, then we have
$$
\mathbb E[\tau(\mathbf G_{\mathrm{simp}},\mathcal P)]=\mathbb E[T_{2,4}]=\frac{291}{56}>\frac{3567}{1120}=\frac{\mathbb E[\mathcal C]}{2}.
$$

\medskip
\noindent
\emph{Case 2: $(q,k)=(3,4)$.}
Here $n=40$, and $\mathbb E[\mathcal C]=\mathbb E[T_{3,4}]=\frac{1621}{351}$. Since two vectors cannot span $\mathbb F_3^4$, we have $\Pr\bigl(\tau(\mathbf{G}_{\mathrm{simp}},\mathcal{P})>t\bigr)=1$, for $0\leq t\leq1$. Denote
$$
\Sigma_{3,4}\triangleq\sum_{t\geq2}\Pr\bigl(\tau(\mathbf G_{\mathrm{simp}},\mathcal P)>t\bigr),
$$
we have $\mathbb E[\tau(\mathbf G_{\mathrm{simp}},\mathcal P)]=2+\Sigma_{3,4}$. It therefore suffices to prove $\Sigma_{3,4}> \frac{\mathbb E[\mathcal{C}]-4}{2} = \frac{217}{702}$.

Suppose first that $n_1\geq3$. In this case, $\bra{k-2}=\bra{2}=4$, $\qbin{k-2}{2}=\qbin{2}{2}=1$. By \eqref{eq:full-jensen-tail}, with $t_0=2$, we have
$$
\sum_{t\geq2}X_t \geq 40 \frac{(1/10)^2}{1-1/10} = \frac{4}{9}.
$$
Moreover, by \eqref{eq:mu-star-combinatorial} $\mu_*\leq \frac{2\cdot2}{n_1(40-n_1)} \leq \frac4{3\cdot37} = \frac{4}{111}$. Thus
$$
3\sum_{t\geq2}Y_t \leq 3 \frac{4/111}{1-4/111} = \frac{12}{107}.
$$
Therefore, by Corollary~\ref{cor:codim-two-bonferroni}, we have
$$
\Sigma_{3,4} \geq \frac{4}{9}-\frac{12}{107} = \frac{320}{963} > \frac{217}{702}.
$$

If $n_1=2$, then \eqref{eq:quotient-bound-in-proof} gives
$$
\mathbb E[\tau(\mathbf G_{\mathrm{simp}},\mathcal P)] \geq \mathbb E[T_{3,2}] = \frac{7}{3} > \frac{1621}{702} = \frac{\mathbb E[\mathcal C]}{2}.
$$
If $n_1=1$,  since $\tau(\mathbf G_{\mathrm{simp}},\mathcal P)$ has the same distribution as $T_{q,k-1}$, we have
$$
\mathbb E[\tau(\mathbf G_{\mathrm{simp}},\mathcal P)]=\mathbb E[T_{3,3}]=\frac{127}{36} > \frac{1621}{702}=\frac{\mathbb E[\mathcal C]}2.
$$
The proof is complete.
\end{IEEEproof}

\section{Proof of Theorem~\ref{thm:random-balanced}}\label{app:random-balanced-proof}

\begin{IEEEproof}
Set $N\triangleq \gv{k}$ and $s\triangleq \frac{N}{M}$. Let $\tau_{\mathrm{rb},k}$ denote the coverage time under the joint randomness of the balanced partition and the subsequent draws. Thus, $\overline E_{\mathrm{rb}}(k,M) = \E[\tau_{\mathrm{rb},k}]$.

Let $X_1,X_2,\ldots$ be i.i.d. columns drawn uniformly from the $N$ columns of $\mathbf{G}_{\mathrm{simp},k}$, and define
$$
\widehat\tau_k \triangleq
\min\left\{ t\ge1: \spn(X_1,\ldots,X_{Mt}) = \mathbb{F}_q^k \right\}.
$$
Clearly, we have $\widehat\tau_k = \left\lceil \frac{\tau_{\mathrm{simp},k}}{M} \right\rceil$.


Fix an integer $1\leq L\leq s$, we regard the indices of the columns drawn in the first $L$ rounds as an ordered array with $ML$ positions. Let $\mu_{\mathrm{rb}}^{(L)}$ and $\mu_{\mathrm{iid}}^{(L)}$ denote the distributions of these arrays in the random balanced model and in the i.i.d. process, respectively. Let $\mathcal{D}_L$ be the event that the $ML$ sampled columns are all distinct. We claim that
\begin{equation}
\mu_{\mathrm{rb}}^{(L)}(\,\cdot\mid\mathcal{D}_L) = \mu_{\mathrm{iid}}^{(L)}(\,\cdot\mid\mathcal{D}_L). \label{eq:rb-same-conditional-law}
\end{equation}
To show this, consider any ordered array $\mathbf{x} =(x_{i,t})_{i\in[M],\,t\in[L]}$
consisting of $ML$ distinct columns. In order to obtain this array, we must have $\{x_{i,t}\}_{t\in [L]}\subseteq  P_i$ holds for all $i\in[M]$ in the random balanced model. The probability that a uniformly random balanced partition satisfies these constraints is $((s)_L)^{M}/(N)_{ML}$, where $(a)_r=a(a-1)\cdots(a-r+1)$ denotes the falling factorial. Conditioned on such a partition, the probability
of drawing the prescribed array is $s^{-ML}$. Hence
$$
\mu_{\mathrm{rb}}^{(L)}(\mathbf{x}) = \frac{(s)_L^M}{(N)_{ML}}s^{-ML},
$$
which is independent of the particular choice of $\mathbf{x}$. Thus, conditioned on $\mathcal{D}_L$, the random balanced process is uniform over all ordered arrays of $ML$ distinct columns.

On the other hand, for the i.i.d. process, every ordered array has probability $N^{-ML}$. Hence, conditioned on $\mathcal{D}_L$, it is also uniform over all ordered arrays of distinct columns. This proves \eqref{eq:rb-same-conditional-law}.

Recall that for probability measures $\mu$ and $\nu$ on a common finite sample space $\Omega$, their total variation distance is defined as
$$
d_{\mathrm{TV}}(\mu,\nu) \triangleq \sup_{A\subseteq \Omega}|\mu(A)-\nu(A)|= \frac12\sum_{\omega}|\mu(\omega)-\nu(\omega)|.
$$
Thus, by \eqref{eq:rb-same-conditional-law}, it follows that
\begin{equation}\label{eq:rb-tv-collision}
    d_{\mathrm{TV}}
    \left( \mu_{\mathrm{rb}}^{(L)},\mu_{\mathrm{iid}}^{(L)} \right) \le
    \Pr_{\mathrm{rb}}(\mathcal{D}_L^c) + \Pr_{\mathrm{iid}}(\mathcal{D}_L^c).
\end{equation}
Under partition $\mathcal{P}_{\mathrm{rb},k}$, columns stored in different container cannot coincide. Thus a collision can occur only between two draws from the same container. By the union bound, we have
$$
\Pr_{\mathrm{rb}}(\mathcal{D}_L^c) \le M\binom{L}{2}\frac{1}{s} = \frac{M^2\binom{L}{2}}{N}.
$$
For the i.i.d. process, we have
$$
\Pr_{\mathrm{iid}}(\mathcal{D}_L^c) \le \binom{ML}{2}\frac{1}{N}.
$$
Therefore,
\begin{equation}
    d_{\mathrm{TV}}
    \left( \mu_{\mathrm{rb}}^{(L)}, \mu_{\mathrm{iid}}^{(L)}
    \right) = O\left(\frac{M^2L^2}{N}\right).
    \label{eq:rb-tv-bound}
\end{equation}

We next control the tails of the two coverage times. 

Let $\beta_k \triangleq \frac{\gv{k-1}}{\gv{k}}$. Consider an arbitrary balanced partition $\mathcal{P}=(P_1,\ldots,P_M)$ of $\bra{k}$ and a hyperplane $H\le\mathbb F_q^k$. Define $\rho_i(H) \triangleq \frac{|P_i\cap H|}{|P_i|}$. Then, we have $\rho_i(H)=M|P_i\cap H|/\bra{k}$ and  $\sum_{i=1}^{M}\rho_i(H) = M|H|/{\gv{k}} = M\beta_k$, which leads to $\prod_{i=1}^{M}\rho_i(H) \le \beta_k^M$ by the AM--GM inequality. Thus, the probability that all columns drawn in the first $t$ rounds lie in $H$ is at most $\beta_k^{Mt}$. Since column set not spanning the whole space $\mathbb{F}_q^k$ is contained in some hyperplane and there are $\gv{k}=N$ hyperplanes, this yields that
\begin{equation}
\Pr(\tau_{\mathrm{rb},k}>t) \le N\beta_k^{Mt}. \label{eq:rb-tail}    
\end{equation}
Moreover, note that $\widehat\tau_k>t$ is equivalent to $X_1,\ldots,X_{Mt}$ lies in certain hyperplane, thus we also have
\begin{equation}
    \Pr(\widehat\tau_k>t) \le N\beta_k^{Mt}. \label{eq:iid-tail}
\end{equation}

For each $0\le t\le L$, let $A_t\subseteq\Omega$ be the event that columns drawn in the first $t$ rounds fail to span $\mathbb F_q^k$. Then, we have $\mu_{\mathrm{rb}}^{(L)}(A_t) = \Pr(\tau_{\mathrm{rb},k}>t)$ and $\mu_{\mathrm{iid}}^{(L)}(A_t) = \Pr(\widehat\tau_k>t)$. By the definition of $d_{TV}(\cdot)$, this further implies that
\begin{align}
    \left| \Pr(\tau_{\mathrm{rb},k}>t) - \Pr(\widehat\tau_k>t)
    \right|
    & = \left|\mu_{\mathrm{rb}}^{(L)}(A_t) - \mu_{\mathrm{iid}}^{(L)}(A_t)
    \right| \nonumber\\
    & \le
    d_{\mathrm{TV}}
    \left( \mu_{\mathrm{rb}}^{(L)},      \mu_{\mathrm{iid}}^{(L)} \right). \label{eq:up-first part}
\end{align}
Recall that $\E[\tau]=\sum_{t\ge0}\Pr(\tau>t)$. Thus, for every $L\ge1$, \eqref{eq:rb-tv-bound}, \eqref{eq:iid-tail} and \eqref{eq:up-first part} imply
\begin{align}
    \left| \E[\tau_{\mathrm{rb},k}] - \E[\widehat\tau_k]
    \right|
    & \le
    \sum_{t=0}^{L}
    \left| \Pr(\tau_{\mathrm{rb},k}>t) - \Pr(\widehat\tau_k>t)
    \right| \nonumber\\
    & \quad + \sum_{t>L}\Pr(\tau_{\mathrm{rb},k}>t) + \sum_{t>L}\Pr(\widehat\tau_k>t) \nonumber\\
    &\le (L+1)O\left(\frac{M^2L^2}{N}\right)
    + \frac{2N\beta_k^{M(L+1)}}{1-\beta_k^M}.
    \label{eq:rb-expectation-comparison}
\end{align}
Take $L=2k$. Since $N=\frac{q^k-1}{q-1}$ and $\beta_k<q^{-1}$,
the first term on the right-hand side of
\eqref{eq:rb-expectation-comparison} is $O_{q}(M^2k^3q^{-k})$, while the second is $O\left(q^{-M-(2M-1)k}\right)$. Consequently, by $M\geq 1$ and $\widehat\tau_k = \left\lceil \frac{\tau_{\mathrm{simp},k}}{M} \right\rceil$, this proves
\begin{align}
    \E[\tau_{\mathrm{rb},k}] & = \E\left[
        \left\lceil \frac{\tau_{\mathrm{simp},k}}{M}
        \right\rceil
    \right] + O_{q,M}(k^3q^{-k}). \label{eq:rb-main-comparison}
\end{align}
Note that $M\left\lceil\frac{c}{M}\right\rceil = c+[-c]_M$ holds for every nonnegative integer $c$. Thus, multiplying \eqref{eq:rb-main-comparison} by $M$ yields
\begin{align}
    M\overline E_{\mathrm{rb}}(k,M) = \E[\tau_{\mathrm{simp},k}] + \gamma_M(k) + o(1). \label{eq:rb-rounding}
\end{align}
Moreover, by Theorem~\ref{thm: simplexE}, we have
\begin{align*}
    \E[\tau_{\mathrm{simp},k}] & = (1-q^{-k})\sum_{j=1}^{k}\frac{1}{1-q^{-j}}  = k+\sum_{j=1}^{k}
    \frac{q^{-j}-q^{-k}}{1-q^{-j}} \\
    & = k+ \sum_{j=1}^{k}
    \frac{1}{q^j-1}
    -
    q^{-k}
    \sum_{j=1}^{k}
    \frac{1}{1-q^{-j}} \\
    & = k+\delta_q+o(1),
\end{align*}
where the last equality follows since $\delta_q=\sum_{j=1}^{k}(q^j-1)^{-1}+\sum_{j> k}(q^j-1)^{-1} = \sum_{j=1}^{k}(q^j-1)^{-1}+O(q^{-k})$ and $q^{-k}\sum_{j=1}^{k}\frac{1}{1-q^{-j}} = O(kq^{-k})$. Hence, \eqref{eq:rb-rounding} becomes
$$
M\overline E_{\mathrm{rb}}(k,M) = k+\delta_q+\gamma_M(k)+o(1).
$$

It remains to analyze $\gamma_M(k)$. Decompose the single-container process into $k$ stages according to the dimension of the space spanned by the drawn columns. For $j\in[k]$, let $G_{k,j}$ denote the number of draws made to have dimension $k-j+1$ conditioned on the current span has dimension $k-j$. Then $G_{k,j}$'s are independent geometric random variables with success probabilities $1-\frac{\gv{k-j}}{\gv{k}}$, and we can write
$$
\tau_{\mathrm{simp},k} = \sum_{j=1}^{k}G_{k,j} = k+D_k,
$$
where $D_k \triangleq \sum_{j=1}^{k}(G_{k,j}-1)$. Then, by $\bra{k}-\bra{k-j}=q^{k-j}\bra{j}$, we have
\begin{align}
    \E[G_{k,j}-1]
    &=
    \frac{\bra{k}}{\bra{k}-\bra{k-j}}-1 \nonumber\\
    &=
    \frac{\bra{k-j}}{\bra{k}-\bra{k-j}} \nonumber\\
    &=
    \frac{1-q^{-(k-j)}}{q^j-1}  \le
    \frac{1}{q^j-1}.
    \label{eq:epochmean}
\end{align}
Note that for each fixed $j$, as $k \to \infty$, $1 - {\bra{k-j}}/{\bra{k}} \to 1 - q^{-j}$. Thus, $G_{k,j} \xrightarrow{d} \Gamma_j$\footnote{For two random variables $X$ and $Y$, the notation $X \xrightarrow{d} Y$ denotes that $X$ converges in distribution to $Y$.}. Moreover, for any fixed finite set of indices $j$, the variables $G_{k,j}$ are independent, hence the joint convergence also holds. Therefore, by $\sum_{j\ge1}\frac{1}{q^j-1}<\infty$, we have $D_\infty=\sum_{j\ge1}(\Gamma_j-1)<\infty$ a.s.

For a fixed $J$, denote
$D_{k,J} = \sum_{j=1}^{J}(G_{k,j}-1)$ and   $D_{\infty,J} = \sum_{j=1}^{J}(\Gamma_j-1)$. Then, since the random variables $\{G_{k,j}\}_{j=1}^J$ are mutually independent and so are $\{\Gamma_j\}_{j=1}^J$, by the Continuous Mapping Theorem (see~\cite[Theorem 2.7]{billingsley13}), we have $D_{k,J} \xrightarrow{d} D_{\infty,J}$.

On the other hand, by \eqref{eq:epochmean}, we have $\sup_k \E[D_k-D_{k,J}] \le \sum_{j>J}\frac{1}{q^j-1}$ and $\E[D_\infty-D_{\infty,J}] = \sum_{j>J}\frac{1}{q^j-1}$, both converge to $0$ as $J\to\infty$. Therefore, for every bounded
Lipschitz function $f$,
\begin{align*}
    \limsup_{k\to\infty}
    \left| \E[f(D_k)]-\E[f(D_\infty)] \right|
    & \le
    \limsup_{k\to\infty}
    \left| \E[f(D_{k,J})]-\E[f(D_{\infty,J})] \right| \\
    & \quad
    + L_f
    \left( \sup_k\E[D_k-D_{k,J}] + \E[D_\infty-D_{\infty,J}] \right).
\end{align*}
The first term on the right-hand side is zero for each fixed $J$, while the
second term tends to zero as $J\to\infty$. Thus,
we have $D_k\xrightarrow{d} D_\infty$. Since $D_k$ and $D_\infty$ are both integer-valued, the weak convergence implies:
\begin{equation}\label{eq:pmflimit}
    \Pr(D_k=d)\longrightarrow \Pr(D_\infty=d),
\end{equation}
for all integer $d\geq 0$.

Now fix $a\in\{0,1,\ldots,M-1\}$ and let $k\to\infty$ with $k\equiv a \pmod M$. For any integer $d\geq 0$, define $r_a(d)\triangleq [-(a+d)]_M$. Since $\tau_{\mathrm{simp},k}=k+D_k$, by the definition of $\gamma_{M}(k)$, we have
$$
\gamma_M(k) = \E[r_a(D_k)] = \sum_{d\ge0}\Pr(D_k=d)r_a(d).
$$
By the definition of $D_k$, we have $\E[D_k] \le \delta_q$ uniformly in $k$ from \eqref{eq:epochmean}. Thus, by Markov's inequality, we have $\Pr(D_k>J)\le \frac{\delta_q}{J}$ and $\Pr(D_\infty>J)\le \frac{\delta_q}{J}$. Note that for each $J\ge0$,
$$
\sum_{d=0}^{J}\Pr(D_k=d)r_a(d) \longrightarrow \sum_{d=0}^{J}\Pr(D_\infty=d)r_a(d)
$$
by the pointwise convergence of the probability mass functions. Moreover, by $r_a(d)\leq M-1$ for every $d\geq 0$, we also have
$$
\left| \sum_{d>J}\Pr(D_k=d)r_a(d)\right| \le (M-1)\Pr(D_k>J) \le (M-1)\frac{\delta_q}{J},
$$
uniformly in $k$, and the same bound holds with $D_\infty$ in place of $D_k$.
Hence, by first letting $k\to\infty$ with $J$ fixed and then letting $J\to\infty$, we obtain
$$
\gamma_M(k) = \E[r_a(D_k)] \longrightarrow \E[r_a(D_\infty)] = \E\!\left[[-(a+D_\infty)]_M\right].
$$
Therefore, along each residue class modulo $M$, the rounding term converges to a limit depending only on that residue class. Equivalently, $\gamma_M(k)$ is asymptotically periodic with period $M$.

Finally, if all $M$ residue classes are attained, we compute the average of these $M$ limiting values $\E[r_a(D_\infty)]$. For each fixed integer $d$, as $a$ ranges from $0$ to $M-1$, the residues $[-(a+d)]_M$ form a permutation of $0,1,\ldots,M-1$. Hence
$$
\frac{1}{M} \sum_{a=0}^{M-1}[-(a+d)]_M = \frac{M-1}{2}.
$$
Then, the result follows directly by averaging over $d=D_\infty$, completing the proof.
\end{IEEEproof}

\bibliographystyle{IEEEtran}
\bibliography{biblio}

\end{document}